\documentclass[11pt]{article}
\usepackage{fullpage}
\usepackage{amsmath,amsthm,amssymb,amsfonts}
\usepackage{color}
\usepackage[shortend]{algorithm2e}
\usepackage{enumitem,comment}

\usepackage{thm-restate}

\usepackage{tasks}

\usepackage{ifthen}
\usepackage{tikz}
\usetikzlibrary{calc}
\usetikzlibrary{intersections}
\usetikzlibrary{shapes.geometric}
\newcommand{\writelabel}[2]{
    \node (TT) at (#1) {#2};
}

\newcommand{\addtext}[1]{#1}

\newcommand{\addsmalltext}[1]{
    \tiny{#1}
}

\newcommand{\addtinytext}[1]{
    \tiny{#1}
}
\newenvironment{radiusgraph}[0]{
    \begin{center}
    \begin{scriptsize}
    \begin{tikzpicture}[
        scale=0.5,
        metric-node/.style={fill, circle, inner sep=0pt, minimum size=2.5pt},
        request-node/.style={fill, diamond, inner sep=0pt, minimum size=4pt, color=trd_main},
        empty-node/.style={circle, inner sep=0pt, minimum size=2.5pt},
        full-node/.style={fill, circle, inner sep=0pt, minimum size=2.5pt},
    ]
}{
    \end{tikzpicture}
    \end{scriptsize}
    \end{center}
}

\newcommand{\definereferencepoints}[1]{
    \coordinate (N2) at ($(#1)+(-45:1.5cm)$);
    \coordinate (N3) at ($(#1)+(-2,-2.5)$);
    \coordinate (N4) at ($(#1)+(-2,-3)$);
}

\newcommand{\drawcircles}[2]{
    \draw [draw=snd_main, fill=snd_fill, fill opacity=0.4] (N1) circle (#1cm);
    \draw [draw=fst_main, fill=fst_fill, fill opacity=0.4] (N2) circle (#2cm);
}

\newcommand{\showradiuses}[0]{
    \draw [style={->}, color=snd_main] (N1) --+(70:2cm) node[midway, right] {$\rho_x$};
    \draw [style={->}, color=fst_main] (N2) --+(160:1cm) node[midway, below] {$\rho_z$};
}

\newcommand{\connectrequests}[0]{
    \draw (N1) -- (N2);
}

\newcommand{\showfirstrequest}[0]{
    \node [request-node, label={[shift={(0.1,-0.05)}] above left:$r_1$}] (R1) at (N1) {};
}

\newcommand{\showsecondrequest}[0]{
    \node [request-node, label={[shift={(-0.1,-0.05)}] above right:$r_2$}] (R2) at (N2) {};
}

\newcommand{\drawtimeline}[1]{
    \draw [->] ($(#1)$) -- ($(#1) + (4,0)$);
    \node at ($(#1) + (4.05,-0.3)$) {\tiny{$t$}};
    
    \foreach \hl in {0,1,2,3} {
        \draw ($(#1) + (\hl,-0.06)$) -- ($(#1) + (\hl,0.06)$);
        \node at ($(#1) + (\hl,-0.35)$) {\tiny{\hl}};
    }
}

\newcommand{\showcurrenttime}[2]{
    \node [metric-node] (T) at ($(#1)+(#2,0)$) {};
}

\newcommand{\drawborder}[2]{
    \draw[rounded corners=4pt, color=gray] (#1) rectangle (#2);
}

\newcommand{\drawplotsegment}[9]{
    \draw [dashed, gray] ($(#1)+(#5,0)$) -- ($(#1)+(#5,#8)$);
    \node (P) at ($(#1)+(#5,#8+0.6)$) {$#9$};
    
    \draw [color=#2] ($(#1)+(#4,#3)$) -- ($(#1)+(#5,#3)$);
    \ifthenelse{\isodd{#6}}{\node [full-node, color=#2] (S) at ($(#1)+(#4,#3)$) {};}{}
    \ifthenelse{\isodd{#7}}{\node [empty-node, draw=#2, fill=white] (E) at ($(#1)+(#5,#3)$) {};}{}
}

\newcommand{\drawsimpleplotsegment}[7]{
    \draw [color=#2] ($(#1)+(#4,#3)$) -- ($(#1)+(#5,#3)$);
    \ifthenelse{\isodd{#6}}{\node [full-node, color=#2] (S) at ($(#1)+(#4,#3)$) {};}{}
    \ifthenelse{\isodd{#7}}{\node [empty-node, draw=#2, fill=white] (E) at ($(#1)+(#5,#3)$) {};}{}
}
\newenvironment{timers}[0]{
    \begin{center}
    \begin{scriptsize}
    \begin{tikzpicture}[
        scale=0.5,
        request-node/.style={fill, diamond, inner sep=0pt, minimum size=4pt, color=trd_main},
    ]
}{
    \end{tikzpicture}
    \end{scriptsize}
    \end{center}
}

\newcommand{\drawarrivaltimeline}[9]{
    \foreach \start/\end/\drawstart in {#9}{
        \draw [rounded corners=2pt, fill=#8, draw=none] ($(#1)+(\start,0.2)$) rectangle ($(#1)+(\end,-0.2)$);
    }

    \draw ($(#1)+(0,-0.15)$) -- ($(#1)+(0,0.15)$);
    \node (L) at ($(#1)+(#2+0.5,0)$) {$#3$};
    
    \foreach \mark/\arrivaltime/\labtype/\point/\number in {#4}{
        \ifthenelse{
            \isodd{\mark}
        }{
            \node [request-node] (M) at ($(#1) + (\arrivaltime,0)$) {};
        }{
            \draw [trd_main, thick] ($(#1)+(\arrivaltime,-0.12)$) -- ($(#1)+(\arrivaltime,0.12)$);
        }
        \node (R) at ($(#1) + (\arrivaltime,0.6)$) {};
    }
    
    \foreach \start/\end/\varname/\varid/\rtype/\rpoint/\rnumber/\isdone in {#5}{
        \ifthenelse{
            \isodd{\isdone}
        }{
            \draw [stealth-stealth] ($(#1)+(\start,0)$) -- ($(#1)+(\end,0)$) node [midway, above] {\ifthenelse{\equal{\rtype}{0}}{$\varname_{\varid}$}{\ifthenelse{\equal{\rtype}{1}}{$\varname_{\varid}^{\rnumber}$}{$\varname_{\varid}^{\rpoint,\rnumber}$}}};
        }{
            \draw [stealth-] ($(#1)+(\start,0)$) -- ($(#1)+(\end,0)$) node [midway, above] {\ifthenelse{\equal{\rtype}{0}}{$\varname_{\varid}$}{\ifthenelse{\equal{\rtype}{1}}{$\varname_{\varid}^{\rnumber}$}{$\varname_{\varid}^{\rpoint,\rnumber}$}}};
            \draw ($(#1)+(\end,-0.12)$) -- ($(#1)+(\end,0.12)$);
        }
    }
    
    \draw [stealth->] ($(#1)+(#6,0)$) -- ($(#1)+(#2,0)$);
    
    \foreach \start/\end/\drawstart in {#9}{
        \ifthenelse{
            \isodd{\drawstart}
        }{
            \draw [stealth-stealth, #7] ($(#1)+(\start,0)$) -- ($(#1)+(\end,0)$);
        }{
            \draw [-stealth, #7] ($(#1)+(\start,0)$) -- ($(#1)+(\end,0)$);
        }
    }
    
    \node (D) at ($(#1) + (#2-0.7,0.6)$) {$\cdots$};
}

\newcommand{\markrequests}[4]{
    \draw [->] (#1) -- +(#2,0);
    \draw ($(#1)+(0,-0.15)$) -- ($(#1)+(0,0.15)$);
    \node (SV) at ($(#1)+(#2+0.5,0)$) {$#3$};
    
    \foreach \mark/\arrivaltime/\labtype/\point/\number in {#4}{
        \ifthenelse{
            \isodd{\mark}
        }{
            \node [request-node] (M) at ($(#1) + (\arrivaltime,0)$) {};
        }{
            \draw [trd_main, thick] ($(#1)+(\arrivaltime,-0.12)$) -- ($(#1)+(\arrivaltime,0.12)$);
        }
        \node (R) at ($(#1) + (\arrivaltime,0.6)$) {\ifthenelse{\isodd{\labtype}}{$r_{\point}^{\number}$}{$r_{\number}$}};
    }
}

\newcommand{\drawgrayarrivaltimeline}[5]{
    \draw [->, gray] (#1) -- +(#2,0);
    \draw [gray] ($(#1)+(0,-0.15)$) -- ($(#1)+(0,0.15)$);
    \node (SV) at ($(#1)+(#2+0.5,0)$) {$#3$};
    
    \foreach \mark/\arrivaltime in {#4}{
        \ifthenelse{
            \isodd{\mark}
        }{
            \node [request-node] (M) at ($(#1) + (\arrivaltime,0)$) {};
        }{
            \draw [trd_main, thick] ($(#1)+(\arrivaltime,-0.12)$) -- ($(#1)+(\arrivaltime,0.12)$);
        }
        \node (R) at ($(#1) + (\arrivaltime,0.6)$) {};
    }
    
    \foreach \start/\end/\varname/\varnum in {#5}{
        \draw [stealth-stealth] ($(#1)+(\start,0)$) -- ($(#1)+(\end,0)$) node [midway, above] {$\varname_{\varnum}$};
    }
    
    \node (D) at ($(#1) + (#2-0.7,0.6)$) {$\cdots$};
}

\newcommand{\drawreferencelines}[3]{
    \foreach \arrival in {#3}{
        \draw [dotted, gray] ($(#1)+(\arrival,0)$) -- ($(#1)+(\arrival,#2)$);
    }
}

\newcommand{\markrequestsandvariables}[6]{
    \draw ($(#1)+(0,-0.15)$) -- ($(#1)+(0,0.15)$);
    \node (SV) at ($(#1)+(#2+0.5,0)$) {$#3$};
    
    \foreach \mark/\arrivaltime/\number in {#4}{
        \ifthenelse{
            \isodd{\mark}
        }{
            \node [request-node] (M) at ($(#1) + (\arrivaltime,0)$) {};
        }{
            \draw [trd_main, thick] ($(#1)+(\arrivaltime,-0.12)$) -- ($(#1)+(\arrivaltime,0.12)$);
        }
        \node (R) at ($(#1) + (\arrivaltime,-0.6)$) {$r_{\number}$};
    }
    
    \foreach \start/\end/\varname/\rtype/\varid/\rid in {#5}{
        \draw [stealth-stealth] ($(#1)+(\start,0)$) -- ($(#1)+(\end,0)$);
        \node (L) at ($(#1)+(\start+1/2*\end-1/2*\start,0.5)$) {\ifthenelse{\isodd{\rtype}}{$\varname_{\varid}^{\rid}$}{$\varname_{\varid}$}};
    }
    
    \draw [stealth->] ($(#1)+(#6,0)$) -- ($(#1)+(#2,0)$);
    
    \node (D) at ($(#1) + (#2-0.7,0.6)$) {$\cdots$};
}

\newcommand{\drawstoppedarrivaltimeline}[9]{
    \foreach \start/\end/\drawstart in {#9}{
        \draw [rounded corners=2pt, fill=#8, draw=none] ($(#1)+(\start,0.2)$) rectangle ($(#1)+(\end,-0.2)$);
    }

    \draw ($(#1)+(0,-0.15)$) -- ($(#1)+(0,0.15)$);
    
    \foreach \mark/\arrivaltime/\labtype/\point/\number in {#4}{
        \ifthenelse{
            \isodd{\mark}
        }{
            \node [request-node] (M) at ($(#1) + (\arrivaltime,0)$) {};
        }{
            \draw [trd_main, thick] ($(#1)+(\arrivaltime,-0.12)$) -- ($(#1)+(\arrivaltime,0.12)$);
        }
        \node (R) at ($(#1) + (\arrivaltime,0.6)$) {};
    }
    
    \foreach \start/\end/\varname/\varid/\rtype/\rpoint/\rnumber/\isdone in {#5}{
        \ifthenelse{
            \isodd{\isdone}
        }{
            \draw [stealth-stealth] ($(#1)+(\start,0)$) -- ($(#1)+(\end,0)$) node [midway, above] {\ifthenelse{\equal{\rtype}{0}}{$\varname_{\varid}$}{\ifthenelse{\equal{\rtype}{1}}{$\varname_{\varid}^{\rnumber}$}{$\varname_{\varid}^{\rpoint,\rnumber}$}}};
        }{
            \draw [stealth-] ($(#1)+(\start,0)$) -- ($(#1)+(\end,0)$) node [midway, above] {\ifthenelse{\equal{\rtype}{0}}{$\varname_{\varid}$}{\ifthenelse{\equal{\rtype}{1}}{$\varname_{\varid}^{\rnumber}$}{$\varname_{\varid}^{\rpoint,\rnumber}$}}};
            \draw ($(#1)+(\end,-0.12)$) -- ($(#1)+(\end,0.12)$);
        }
    }
    
    \foreach \start/\end/\drawstart in {#9}{
        \ifthenelse{
            \isodd{\drawstart}
        }{
            \draw [stealth-stealth, #7] ($(#1)+(\start,0)$) -- ($(#1)+(\end,0)$);
        }{
            \draw [-stealth, #7] ($(#1)+(\start,0)$) -- ($(#1)+(\end,0)$);
        }
    }
}
\definecolor{fst}{RGB}{105,190,40}
\colorlet{fst_main}{fst!90!black}
\colorlet{fst_fill}{fst!10}

\definecolor{snd}{RGB}{13,152,186}
\colorlet{snd_main}{snd!90!black}
\colorlet{snd_fill}{snd!10}

\definecolor{add}{RGB}{177,148,216}
\colorlet{add_main}{add!90!black}
\colorlet{add_fill}{add!10}

\definecolor{trd}{RGB}{255,210,0}
\colorlet{trd_main}{trd!90!black}

\def\E{\mathbb{E}}

\def\R{\mathcal{R}}
\def\S{\mathcal{S}}

\newcommand{\optim}{{\rm OPT}}
\newcommand{\algor}{{\rm ALG}}
\newcommand{\cc}{{\rm CC}}
\newcommand{\dc}{{\rm DC}}

\newcommand{\roe}{{\rm RoE}}
\newcommand{\prob}{\mathbb{P}}
\newcommand{\greedy}{{\rm Greedy}}
\newcommand{\rad}{{\rm Radius}}

\allowdisplaybreaks[4]

\newtheorem{definition}{Definition}

\newcommand{\metricspace}{\mathcal{M}}
\newcommand{\setofpoints}{\mathcal{X}}
\newcommand{\anypoint}{x}

\newcommand{\distancesymbol}{d}
\newcommand{\distance}[2]{\distancesymbol(#1,#2)}

\newcommand{\location}[1]{\ell(#1)}
\newcommand{\arrival}[1]{t(#1)}
\newcommand{\request}[1][]{r_{#1}}

\newcommand{\waitingparam}[1][]{\lambda_{#1}}

\newcommand{\radius}[1]{\rho_{#1}}

\newcommand{\positivesubset}[1]{#1^{+}}
\newcommand{\realnum}{\mathbb{R}}

\newcommand{\expdistr}[1]{{\rm Exp}(#1)}

\newtheorem{theorem}[definition]{Theorem}
\newtheorem{lemma}[definition]{Lemma}
\newtheorem{observation}[definition]{Observation}
\newtheorem{proposition}[definition]{Proposition}
\newtheorem{corollary}[definition]{Corollary}

\newtheorem{claim}[definition]{Claim}

\usepackage[textsize=footnotesize,textwidth=2cm,color=green!50!gray]{todonotes} 

\title{Online matching with delays and stochastic arrival times}
\author{Mathieu Mari\thanks{LIRMM, University of Montpellier, Montpellier, France. mari.mathieu.06@gmail.com}, Michał Pawłowski\thanks{MIMUW, University of Warsaw and IDEAS NCBR, Warsaw, Poland. michal.pawlowski196@gmail.com}, Runtian Ren\thanks{Institute of Computer Science, University of Wrocław, Wrocław, Poland. renruntian@gmail.com}, Piotr Sankowski\thanks{MIMUW, University of Warsaw and MIM Solutions, Warsaw, Poland. piotr.sankowski@gmail.com}\thanks{A preliminary version appeared in the Proceedings of the 22nd International Conference on Autonomous Agents and Multi-agent Systems (AAMAS) 2023 pp. 976–984.}}
\date{}

\begin{document}

\maketitle

\begin{abstract}
This paper presents a new research direction for the Min-cost Perfect Matching with Delays (MPMD), a problem introduced by Emek et al.\ (STOC'16). 
In the original version of this problem, we are given an $n$-point metric space, where requests arrive in an online fashion. 
Our goal is to minimize the matching cost for an even number of requests. 
However, contrary to traditional online matching problems, a request does not have to be paired immediately at the time of its arrival. 
Instead, the decision of whether to match a request can be postponed for time $t$ at a delay cost of $t$. 
For this reason, the goal of the MPMD is not only to minimize the distance cost of the generated matching but to minimize the overall sum of distance and delay costs. 
Interestingly, it is proved that in the standard case of the adversarially generated requests, no online algorithm can achieve a competitive ratio better than $O(\log n / \log \log n)$ (Ashlagi et al., APPROX/RANDOM'17).

Here we consider a stochastic version of the MPMD problem where the input requests follow a Poisson arrival process. 
For such problem, we show that the above lower bound can be improved by presenting two deterministic online algorithms which, in expectation, are constant competitive, i.e., the ratio between the expected costs of the output matching and the optimal offline solution is bounded by a constant. 
The first one is a simple greedy algorithm that matches any two requests once the sum of their delay costs exceeds their connection cost, i.e., the distance between them. 
The second algorithm builds on the tools used to analyze the first one in order to obtain even better performance guarantees. 
This result is rather surprising as the greedy approach cannot achieve a competitive ratio better than $O(m^{\log 1.5 + \varepsilon})$ in the adversarial model, where $m$ denotes the number of agents. 
Finally, we prove that it is possible to obtain similar results for the general case when the delay cost follows an arbitrary positive and non-decreasing function, for the asymmetric distance case, as well as for the MPMD variant with penalties to clear pending requests.
\end{abstract}


\section{Introduction}
Imagine players logging into an online platform to compete against each other in a two player game. 
The platform needs to pair them in a way that maximizes the overall satisfaction from the gameplay. 
Typically, a player prefers to be matched with someone with similar gaming skills. 
Thus, the platform has to consider the experience gap when pairing two players. 
This skill level difference is referred to as the \emph{connection cost}. 
Additionally, once logged in, a player can tolerate some waiting time to be matched --- this is why the platform can postpone the pairing decision in the hope of a better matching to be found (i.e., the login of another player with similar skills). 
Nonetheless, the waiting time for each player has its limits. 
A player may become unsatisfied if their gaming request has been ignored for too long. 
This time gap between logging into the platform and joining a gaming session is referred to as the \emph{delay cost}. 
The platform's goal is to pair all the online players into sessions, such that the total connection cost plus the total delay cost produced is minimized.

The above is an example of an online problem called Min-cost Perfect Matching with Delays (MPMD) \cite{emek2016online}. 
It has drawn researchers attention recently \cite{emek2016online, azar2017polylogarithmic, ashlagi2017min, bienkowski2017match, bienkowski2018primal, liu2018impatient, azar2020deterministic, azar2021min} due to many real-life applications ranging from Uber rides, dating platforms, kidney exchange programs etc. 
Formally, the problem of MPMD is defined as follows. 
The input is a set of $m$ requests arriving at arbitrary times in a metric space $\metricspace = (\setofpoints, \distancesymbol)$ equipped with a distance function $d$.
Here, $m$ is an even integer, and $\setofpoints$ denotes the set of points in $\metricspace$. 
Each request $\request$ is characterized by its \emph{location} $\location{\request} \in \setofpoints$ and \emph{arrival time} $\arrival{\request} \in \positivesubset{\realnum}$.
When two requests $\request$ and $\request'$ are matched into a pair at time $t \ge \max\{\arrival{\request}, \arrival{\request'}\}$, a \emph{connection cost} $\distance{\location{\request}}{\location{\request'}}$ plus a \emph{delay cost} $(t - \arrival{\request}) + (t - \arrival{\request'})$ is incurred.
The target is to minimize the total cost produced by the online algorithm for matching all the requests into pairs.

Previously, the MPMD problem was studied in an adversarial model where an online adversary generated the requests at different times in the given metric space $\metricspace$.
Under this adversarial model, no online algorithm can achieve a constant competitive ratio:
\begin{itemize}
    \item[-] if the metric is known in advance, the current best competitiveness is $O(\log n)$ (here $n$ denotes the number of points in the metric) \cite[Theorem 3.1]{azar2017polylogarithmic} and no online algorithm can achieve competitive ratio better than $\Omega(\log n / \log \log n)$ \cite[Theorem 1]{ashlagi2017min};
    \item[-] if the metric is not known in advance, the current best competitiveness is $O(m^{\log 1.5 + \varepsilon}/ \varepsilon)$ (with $\varepsilon > 0$), achieved tightly by a deterministic online greedy algorithm \cite[Theorem 1]{azar2020deterministic}. 
\end{itemize}

In fact, it is often too pessimistic to assume no stochastic information on the input is available. 
Again, consider the example of matching gaming requests. 
The online gaming platform has all the historical data and can estimate the arrival frequency of the players with each particular skill level on an hourly basis. 
Therefore, it is reasonable to assume that the gaming requests follow some stochastic distribution. 
Depending on the time of day, though, there may be more or fewer players logging in. 
However, if we divide the timeline into small intervals, it is reasonable to assume that within each of them, the distribution is regular and the requests are mutually independent (since the players don't know each other).
Based on these observations, the following question can be naturally stated: 
{\em in the case when stochastic information on the input is available, can we devise online algorithms for MPMD with better performance guarantees?}

In this paper, we provide an affirmative answer to the question above. 
We consider a stochastic online version of MPMD, by assuming that the requests arrive following a Poisson arrival process. 
To be more precise, the waiting time between any two consecutive requests arriving at any metrical point $\anypoint$, follows an exponential distribution $\expdistr{\waitingparam[\anypoint]}$ with parameter $\waitingparam[\anypoint] \ge 0$.
Under such a model, the goal of the platform is to minimize the expected cost produced by an algorithm $\algor$ to deal with a random input sequence consisting of $m$ requests.
To evaluate the performance of our algorithms on stochastic inputs, we use the {\em ratio of expectations}, that corresponds to the ratio of the expected cost of the algorithm to the expected cost of the optimal offline solution (see Definition \ref{def:roe}).


\paragraph{Our contribution.}
We prove that the performance guarantee obtained in the Poisson arrival model is significantly better compared with the current best competitiveness obtained in the adversarial model.
More specifically, we show that an intuitive \emph{Greedy} algorithm, which matches any two requests immediately when their total delay cost reaches their distance, achieves a constant ratio of expectations.
\begin{restatable}{theorem}{greedyapprox}
\label{main:greedy}
For MPMD in the Poisson arrival model, the Greedy algorithm achieves a ratio of expectations of $16 / (1 - e^{-2})$.
\end{restatable}
To prove this theorem, we apply the following strategy. 
We first notice that the connection cost of a Greedy solution is at most its delay cost. 
Thus, it becomes the core of the proof to upper bound the delay cost. 
For this purpose, in Section \ref{section:algorithms}, we define the \emph{radius} $\rho_x \ge 0$ for each metric point $x$. Such a radius depends on the parameters of the problem and roughly corresponds to the expected delay time for matching the requests located on $x$. 
Then, we show how to use the radius to lower bound the cost of the optimal offline solution.
Intuitively, we prove that a request located on $x$ is in expectation responsible for a total cost of $\Omega(\rho_x)$. 
At this point, it is worth emphasizing once again that in the adversarial model when the metric is not known in advance the current best known competitive ratio is $\Omega(m^{\log \frac{3}{2} + \varepsilon})$) (see the counter example in \cite{azar2020deterministic}, Appendix A).

This notion of radius suggests another potential algorithm for MPMD with stochastic inputs. 
Indeed, when a new request $r$ arrives on a point $x$, we know that this request will wait for a time $O(\rho_x)$ in average before being matched by the Greedy algorithm. 
In particular, $r$ will be matched with another request that is at distance $O(\rho_x)$. 
Therefore, if at the time of the $r$'s arrival, there is another pending\footnote{By pending we mean that at that time, the request is still unmatched by the algorithm.} request $r'$ that is at distance less than $\rho_x$, why not match these two requests directly? 
In Section \ref{section:algorithms}, we formalize this intuition and design an algorithm called \emph{Radius}. 
Thanks to these anticipated pairings, the performance ratio is improved by a factor of 2.  
\begin{restatable}{theorem}{radiusapprox}
\label{main:radius}
For MPMD in the Poisson arrival model, the Radius algorithm achieves a ratio of expectations of $8 / (1 - e^{-2})$.
\end{restatable}

Finally, we show how to adjust the Greedy and the Radius algorithms to deal with other variants of the MPMD problem in a way that preserves constant performance ratio. 
In Section \ref{sec:general_delay}, we look at the generalization of the problem where a request can be delayed for a time $t$ at a cost $f(t)$, where $f$ is a given positive and non-decreasing function. 
We show that, unless $f$ is such that the expected cost of the optimal offline solution is infinite, our algorithms achieve constant performance ratios, where the constants only depend on the delay cost function $f$. 
In Section \ref{section:mpmdfp}, we consider the variant of MPMD where we are allowed to clear pending requests for a fixed penalty cost. 

\paragraph{Related work.}
The MPMD problem was introduced by Emek et al.\ \cite{emek2016online}. 
In their paper, they proposed a randomized online algorithm that achieves a competitive ratio of $O(\log^2 n + \log \Delta)$, where $n$ is the number of points of the metric space and $\Delta$ is the aspect ratio.  
Later, Azar et al.\ \cite{azar2017polylogarithmic} improved the competitive ratio to $O(\log n)$, thereby removing the dependence of $\Delta$ in the competitive ratio. 
Both of these papers randomly embed the metric space into a tree of distortion $O(\log n)$, and then propose online algorithms on tree metrics. 
In the adversarial model, this bound is essentially tight, since Ashlagi et al.\ \cite{ashlagi2017min} showed that any randomized algorithm achieves a competitive ratio of $\Omega(\log n / \log \log n)$. 
Note that the above results assume that the $n$-point metric is given in advance. 
When the metric is not known in advance, Bienkowski et al.\ proposed a $O(m^{2.46})$-competitive online greedy algorithm \cite{bienkowski2017match} and a $O(m)$-competitive online algorithm based on the primal-dual method \cite{bienkowski2018primal}, where $m$ denotes the number of requests released. 
Azar and Jacob-Fanani \cite{azar2020deterministic} later proposed a $O(m^{\log 1.5 + \varepsilon}/ \varepsilon)$-competitive greedy algorithm, which is currently the best deterministic online algorithm.
In the special case of a two-points metric, Emek et al.\ \cite{emek2019minimum} proposed a 3-competitive deterministic greedy algorithm, and He et al.\cite{he2023randomized} proposed a 2-competitive randomized online algorithm. 
Deryckere and Umboh \cite{deryckere2023online} studied online matching with set delay, where the delay cost at any given time is an arbitrary function of the set of pending requests.

Another line of work considered a bipartite variant of MPMD, i.e., the Min-cost Bipartite Perfect Matching with (linear) Delays (MBPMD), where each request can be either red or blue, and only two requests of different colors can be matched into a pair.
For MBPMD, Ashlagi et al.\ \cite{ashlagi2017min} presented two algorithms achieving a competitive ratio of $O(\log n)$ --- the first is an adaptation of Emek et al.'s \cite{emek2016online} algorithm to the bipartite case, and the second is an adaptation of the algorithm proposed by Azar et al.\ \cite{azar2017polylogarithmic}.  
Besides, Ashlagi et al.\ \cite{ashlagi2017min} presented a lower bound of $\Omega(\sqrt{\log n / \log \log n})$ on any randomized algorithm.
Recently, inspired from the Robust Matching (RM) algorithm proposed by Raghvendra \cite{raghvendra2016robust, nayyar2017input} for the minimum cost bipartite perfect matching (MBPM) problem, Kuo \cite{kuo2024online} proposed a better $O(m^{0.5} \cdot \log^2 m)$-competitive online algorithm, which is currently the best deterministic online algorithm for MBPMD. 

The MPMD and MBPMD problems have been investigated in the more general case when any request can be delayed for a duration $t$ at a cost $f(t)$. 
Liu et al.\ \cite{liu2018impatient} considered the case when $f$ is a convex function, and established a lower bound $\Omega(n)$ on the competitive ratio of any deterministic algorithm for Convex-MPMD. 
Specifically, this bound is obtained for an $n$-point uniform metric and a delay function of the form $f(t) = t^{\alpha}$ for $\alpha > 1$. 
In this case, they presented a deterministic algorithm that achieves a competitive ratio of $O(n)$. 
In the case when $f$ is a concave function, Azar et al.\ \cite{azar2021min} gave a $O(1)$-competitive (resp.\ $O(\log n)$-competitive) deterministic online algorithms for MPMD and MBPMD for a single-point metric (resp.\ any metric). 

Other classical online problems have been also considered under such delay setting, such as online service \cite{azar2017online, bienkowski2018online, azar2019general, touitou2023improved}, multi-level aggregation \cite{bienkowski2016online, buchbinder2017depth, carrasco2018online, azar2019general, bienkowski2021new, le2023power, mari2024online}, facility location \cite{bienkowski2021online,azar2019general,azar2020beyond}, bin packing \cite{azar2019price, epstein2021bin}, set cover \cite{azar2020set,touitou2021nearly, le2023power} and many others \cite{melnyk2021online, gupta2020caching, azar2020beyond, touitou2021nearly, chen2022online, kakimura2023deterministic, kawase2024online}.

One stochastic online (weighted) matching problem \cite{collina2020dynamic, aouad2020dynamic, kakimura2021dynamic, kessel2022stationary, baumler2023superiority} can be seen as a deadline variant of our problem.
That is, each request of a particular type $i$ arrives with a Poisson arrival rate $\lambda_i$; after arrival, this request departs with a Poisson rate $\mu_i$. 
Matching a type $i$ request with a type $j$ request creates a value $v_{ij}$ and the target is to maximize the total value of the matching solution produced online.
Collina et al. \cite{collina2020dynamic} proposed a randomized online algorithm based on linear programming, which achieves a ratio of expectations of 1/8. 
Aouad and Sarita{\c{c}} \cite{aouad2020dynamic} proposed better algorithm with ratio-of-expectation of $(1 - e^{-1}) / 4$. 
Kessel et al. \cite{kessel2022stationary} studied the bipartite version of the problem and B{\"a}umler et al. \cite{baumler2023superiority} studied the special case to maximize the number of pairs produced. 
Finally, we remark that matching is a huge topic, drawing attentions from both the theory and real applications perspectives since the seminal work of Edmonds \cite{edmonds1965maximum, edmonds1965paths}. 
In recent years, motivated by job market, kidney exchanges etc, many other online matching results have also been conducted, e.g., \cite{farhadi2022generalized, goyal2022secretary, boehmer2022proportional, cho2022two, brilliantova2022fair, ma2022group, kamiyama2020stable, kawase2020approximately, zhou2019robust, kern2019generalized, aziz2017stable, brubach2017attenuate, pini2014stable}.
Different from MPMD, these works assume that the matching decision must be made immediately at request arrival. 
To the best of our knowledge, we are the first to consider MPMD in the stochastic arrival model. 


\paragraph{Paper Organization.}
We first introduce all the necessary preliminaries in Section \ref{section:preliminaries}. 
Next, we present the details of how Greedy and Radius algorithms work in Section \ref{section:algorithms}.
In Sections \ref{section:lowerbound} - \ref{section:proof_greedy}, we provide the proofs for the lemmas mentioned in Section \ref{section:algorithms} correspondingly. 
In Section \ref{sec:general_delay}, we consider the general case where the delay cost follows an arbitrary positive and non-decreasing function.
In Section \ref{section:mpmdfp}, we consider a variant of MPMD called MPMDfp, where it is allowed to clear a request by paying a penalty. 
In Section \ref{section:asymmetric}, we consider an asymmetric distance case where the distance cost between two requests located at $x, y$ is $(d(x, y)+d(y, x)) / 2$.
Finally, in Section \ref{section:conclusion}, we mention some concluding remarks and open questions. 

\section{Preliminaries}
\label{section:preliminaries}
\paragraph{Problem statement.}
A \emph{metric space} $\metricspace = (\setofpoints, \distancesymbol)$ is a set of points $\setofpoints$ equipped with a distance function $\distancesymbol: \setofpoints \times \setofpoints \to \positivesubset{\realnum}$ that satisfies the triangle inequality. 
The input for the MPMD problem consists of a sequence $\sigma$ of $m$ requests ($m$ being an even integer), where each request $\request \in \sigma$ is characterized by its \emph{location} $\location{\request} \in \setofpoints$ and \emph{arrival time} $\arrival{\request} \in \realnum^+$ (w.l.o.g., suppose that no two requests arrive at the same time). 
Now, given any solution for an input sequence $\sigma$, let $M$ denote the set of paired requests (i.e., the perfect matching generated for $\sigma$), and let $s(\request) \ge t(\request)$ denote the moment when a request $\request$ is matched. 
Note that if $\request$ and $\request'$ are matched into a pair, i.e., $(\request, \request') \in M$, we have $s(\request) = s(\request')$. 
Using this notation, the total cost of a solution $(M, s)$ is the sum of its \emph{delay cost} and its \emph{connection cost} defined as follows. 
The delay cost produced by the solution is the sum of the delay costs $s(\request) - \arrival{\request}$ incurred for each request $\request$. 
Similarly, the connection cost is the sum of distances between all the paired requests, i.e., $\sum_{(r, r') \in M} d(\ell(r), \ell(r'))$.

Let $\optim(\sigma)$ denote the minimum cost of a feasible solutions for $\sigma$. 
Notice that it corresponds to a minimum weight perfect matching for $\sigma$, where the weight of an edge $(r, r') \in \sigma \times \sigma$ is given by $d(\ell(r),\ell(r')) + |t(r) - t(r')|$. Indeed, for any pair ($r, r'$) produced by the optimal solution it holds that $s(r) = s(r') = \max\{t(r), t(r')\}$. 
This observation implies that the optimal offline solution can be computed in polynomial time of the number of requests.
In this paper, we are interested in the design of \emph{online} algorithms for the problem: the decision of matching a pair $(r,r')$ at time $t$ only depends on $\{r \in \sigma: t(r) \le t\}$, and this decision is irrevocable. 

\paragraph{Stochastic model.}
In the stochastic version of MPMD, the goal is to design an online algorithm that processes a sequence of requests arriving at ``random moments'', instead of being generated by an online adversary. 
To formalize the notion of random arrival times, we use the Poisson arrival process: given any point $\anypoint \in \setofpoints$, we assume that the requests arrive at $\anypoint$ with a Poisson arrival rate $\waitingparam[\anypoint] > 0$. 
Recall that an exponential variable $X \sim \expdistr{\waitingparam}$ with parameter $\lambda>0$ has a probability density function $f_{\waitingparam}(t) = \waitingparam e^{-\waitingparam t}$ for $t \ge 0$ and expectation $\E[X] = 1 / \waitingparam$. 
The exponential distribution may be viewed as a continuous counterpart of the geometric distribution, which describes the number of Bernoulli trials necessary for a discrete process to change state (here, observing a new request on a given point). 
The exponential distribution is used for instance in physics to model the time until a radioactive particle decays, or in queuing theory to model the time it takes for an agent to serve the request of a customer. 

\begin{definition}[distributed Poisson arrival model]
We say that a (random) \emph{requests sequence} $\sigma$ follows distributed Poisson arrival model if the \emph{waiting time} between any two consecutive requests arriving at the same point $\anypoint \in \setofpoints$ follows an \emph{exponential distribution} with parameter $\waitingparam[\anypoint] > 0$ and the variables representing waiting times are mutually independent.
\end{definition}

In this paper, when we say that $\sigma$ is \emph{a random request sequence of length} $m$, for some integer $m$, we mean that $\sigma$ consists of $m$ requests, and that the arrival times of the requests follow the above distributed Poisson arrival model. 
In this context we measure the performance of our algorithms using \emph{ratio of expectations}:
\begin{definition} \label{def:roe}
We say that an algorithm $\algor$ for MPMD has a ratio of expectations $C \ge 1$, if 
\begin{equation*}
    \overline{\lim_{m \to \infty}} \frac{\E_{\sigma}^m[\algor(\sigma)]}{\E_{\sigma}^m[\optim(\sigma)]} \le C,
\end{equation*}
where $\algor(\sigma)$ (resp.\ $\optim(\sigma)$) denotes the cost of $\algor$ (resp.\ an optimal offline solution) on the request sequence $\sigma$, and $\E_{\sigma}^m[\algor(\sigma)]$ (resp.\ $\E_{\sigma}^m[\optim(\sigma)]$) denotes the expected cost of $\algor$ (resp.\ $\optim$) on a random sequence $\sigma$ consisting of $m$ requests generated by the Poisson arrival process. 
\end{definition}

Note that there are other criteria to measure the performance of online algorithms in the stochastic input model, such as \emph{expectation of ratio}, defined as the expected value of the ratio $\algor(\sigma)/\optim(\sigma)$ over all random inputs $\sigma$ (see e.g. \cite{garg2008stochastic}).\footnote{To clarify, the criteria for stochastic input model, no matter ratio of expectations or expectation of ratio, is ``weaker'' than the classic competitive ratio, which is used to evaluate the performance of an online algorithm in the adversarial model. 
This is because, each online optimization problem can be interpreted as a game between the online algorithm and the adversary: the adversary releases a sequence of requests to ``challenge'' the online algorithm, and the online algorithm has to make decisions during input being revealed incrementally.
In the classic adversarial online model, the adversary can release the vicious requests based on the decisions made by the online algorithm, for the purpose of making the online algorithm produce ``bad'' results compared with the optimal offline solution. 
However, in the stochastic model, the adversary has to obey some statistic rules (such as Poisson property on the requests arrivals), and hence is not able to release the online requests arbitrarily. As a result, for any online problem, the performance ratio of an online algorithm in the stochastic model, is typically better than the competitive ratio of this algorithm in the adversarial model.} 

We now present a more analysis-friendly version of the Poisson arrival model, referred to as the \emph{centralized} model. 
To design this equivalent process, we exploit the two well-known properties of the exponential distribution. 

\begin{proposition}[memoryless property] \label{proposition:memoryless}
    If $X$ is an exponential variable with parameter $\waitingparam$, then for all $s,t\ge 0$, we have
    \begin{equation*}
        \prob\left(X > s + t \ |\  X > s\right) = \prob(X > t) = e^{-\waitingparam t}.
    \end{equation*}
\end{proposition}

\begin{proposition} \label{proposition:minimum}
    Given $n$ independent exponential variables $X_i \sim \expdistr{\lambda_i}$ for $i \in \{1, 2, \dots, n\}$, let $Z := \min\{X_1, X_2, \dots, X_n\}$ and let $\lambda := \sum_{i = 1}^n \lambda_i$. It holds that
    \begin{tasks}[style=enumerate](3)
        \task $Z \sim \expdistr{\lambda}$,
        \task $\prob(Z = X_i) =  \lambda_i / \lambda$,
        \task $Z \perp \{Z = X_i\}$,
    \end{tasks}
    where $\perp$ denotes independence.
\end{proposition}

To construct the centralized model, we first assign to each point $\anypoint$ a clock with a timer that goes off at the moment determined by an exponential variable with parameter $\waitingparam[\anypoint]$. 
We define the timers to set themselves immediately after they ring and start counting the next exponentially-distributed time interval. 
Here, we assume that all timers are mutually independent and independent of the history prior to the time they were set. 
Given this setup, we say that a request arrives at point $\anypoint$ every time its timer goes off.

\begin{figure}[ht]
\input{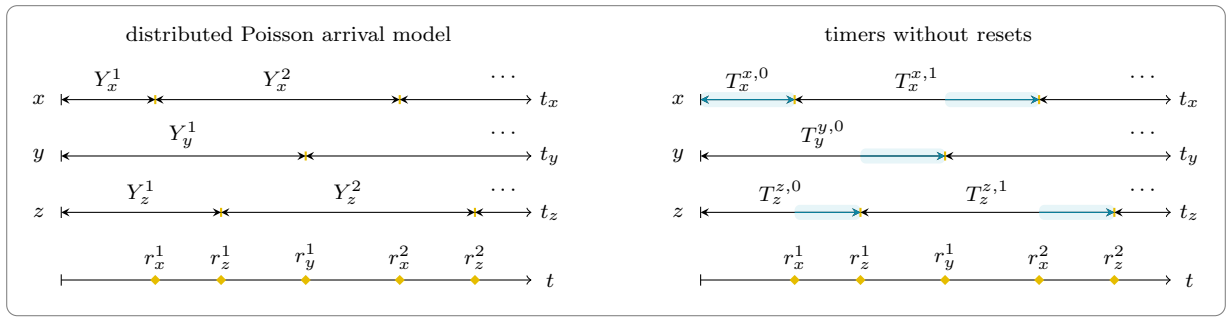}
\vspace{-16pt}
\caption{Example showing the correspondence between distributed Poisson arrival model and exponential timers (without resets). The graph on the right highlights in blue the waiting times between the consecutive arrivals from the perspective of the whole metric space.}
\label{figure:distributed_Poisson}
\end{figure}

It is easy to notice that as of now, we have only rephrased the description of the previous model (see Figure \ref{figure:distributed_Poisson}, where $Y_x^{i+1}$ and $T_x^{x, i+1}$ both represent the waiting time between the $i$-th and the $(i+1)$-th request arriving at $x$). 
Timers, though, allow us to reset them at any point (i.e. stop the current timer and set a new one), obtaining an equivalent stochastic process. 
Indeed, by the memoryless property, resetting running timers is equivalent to letting them continue to run. 
Using this insight, we can prove that the following arrival model is equivalent to the one already presented.

\begin{definition}[centralized Poisson arrival model] \label{definition:centralized_poisson_arrival}
    We say that a (random) \emph{requests sequence} $\sigma$ follows centralized Poisson arrival model if the \emph{waiting time} between any two consecutive requests in the given metric space follows an \emph{exponential distribution} with parameter $\waitingparam(\setofpoints) := \sum_{\anypoint \in \setofpoints} \lambda_\anypoint$ and each time a request arrives, the probability of it appearing at point $\anypoint$ equals $\waitingparam[\anypoint] / \waitingparam(\setofpoints)$. 
    We assume that the waiting times and requests location choices are all mutually independent.
\end{definition}

To get a better understanding of how this model relates to timers with resets, see Figure \ref{figure:centralized_Poisson}. 
The graph on the left shows the centralized model where the waiting time between appearances of the $i$-th and the $(i+1)$-th requests is determined by the realization of variable $Y_{i+1} \sim \expdistr{\waitingparam(\setofpoints)}$. 
On the right, we set the timers $T_\anypoint^i$ for each point $\anypoint$ independently and wait for their minimum to go off. 
By Proposition \ref{proposition:minimum}, the minimum follows an exponential distribution with parameter $\expdistr{\waitingparam(\setofpoints)}$ and has the same appearance distribution over points as the centralized model.

\begin{figure}[ht]
\input{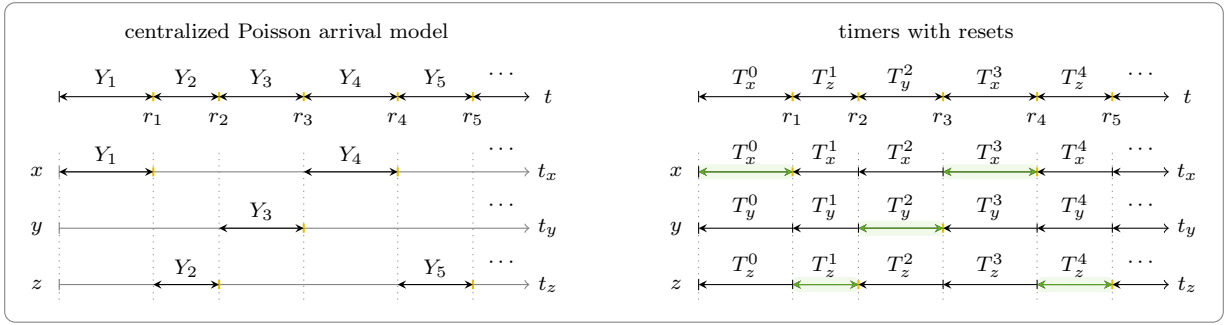}
\vspace{-16pt}
\caption{Example showing the correspondence between the centralized Poisson arrival model and timers with resets. On the right, a double-headed arrow represents a timer that went off, while a single-headed arrow means that we had to reset a timer.}
\label{figure:centralized_Poisson}
\end{figure}

Notice that since both models are equivalent, it gives us another way of looking at the stochastic process we work with --- it is sufficient to define an arrival rate for the whole metric space and adjust the requests appearance distribution over the points.

\section{Constant competitive algorithms}
\label{section:algorithms}
In this section, we introduce two deterministic online algorithms for the MPMD problem: Greedy and Radius. 
We formally define the radius of each metric point which is used to design the Radius algorithm. 
We present the upper bounds on the expected cost of our algorithms (Lemmas \ref{lemma:greedy_upperbound} and \ref{lemma:radius_upperbound}) and the lower bound on the expected cost of the optimal offline solution (Lemma  \ref{lemma:lowerboundingscheme}). 
We give an overview of the techniques used to obtain these bounds. 
Finally, with these Lemmas, we prove Theorems \ref{main:greedy} and \ref{main:radius}. 

\subsection{The Greedy algorithm}
First, let us present a simple greedy algorithm. 
Its strategy is pretty straightforward: once the total waiting time of any two pending requests exceeds the distance between them, Greedy immediately matches them together. 
It is easy to show that this algorithm is well-defined. 
Indeed, since the metric space $\metricspace$ is bounded (as it contains a finite number of points), the waiting time of the last request is bounded by the diameter of $\metricspace$. 
Together with the assumption that the input sequence $\sigma$ has an even number of requests, it proves that Greedy outputs a perfect matching on $\sigma$.
Notice that this algorithm works more generally in the online adversarial model, and additionally that it does not require to know the metric space or the exponential parameters in advance. 
For a more formal description of this greedy procedure, see the pseudo-code of Algorithm \ref{pseudocode:greedy_original}.\\

\RestyleAlgo{boxruled}
\LinesNumbered
\SetAlgoVlined
\begin{algorithm}
\caption{Greedy}
\label{pseudocode:greedy_original}
\KwIn{A sequence $\sigma$ of requests.}
\KwOut{A perfect matching of the requests.}
\For{any time $t$}{
    \If{
        there exist pending requests $\request, \request'$ such that $(t - t(\request)) + (t - t(\request')) \ge \distance{\location{\request}}{\location{\request'}}$
    }{
        match them into a pair with ties broken arbitrarily.
    }
}
\end{algorithm}
\vspace{-6pt}

\begin{figure}[ht]
\input{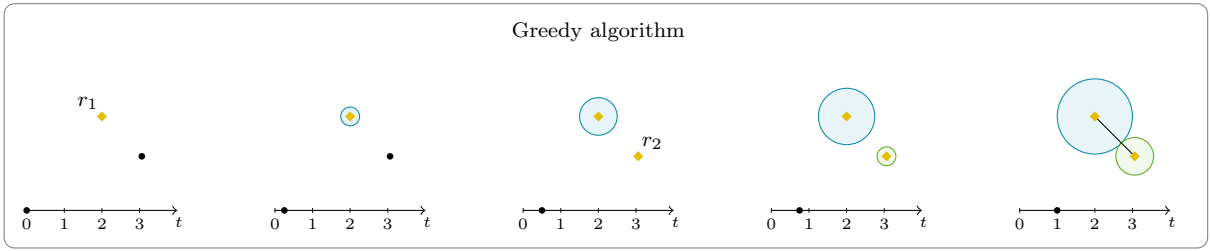}
\vspace{-16pt}
\caption{An example of how Greedy works on a sequence of two requests arriving at times 0 and 0.5 in a 2-point metric space with the distance between the points equal to 1.5.}
\label{figure:greedy_algorithm}
\end{figure}

To better understand the algorithm, we can look at its geometrical interpretation. 
Here, when a request $\request$ appears at some point $\anypoint$, a ball centred at $x$ starts growing with a uniform rate as time passes by. 
The radius of this ball represents the delay cost incurred due to leaving $\request$ unmatched. 
Hence, once two balls intersect, the pending requests located at their centres are paired. 
Figure \ref{figure:greedy_algorithm} shows an example of such a process.

The remaining part of this subsection presents a sketch of how to prove the constant ratio of expectations for Greedy (Theorem \ref{main:greedy}). 
First, we observe that for each request served by this algorithm, its connection cost does not exceed its delay cost. 
Thus, if we find the upper bound for the latter, we will be able to estimate the total expected cost of the matching generated by Greedy on a request sequence $\sigma$. 
To do so, let us focus on finding the expected delay cost of a single request $\request$ arriving at some point $\anypoint \in \setofpoints$. 
We say that it is matched with a close request if the distance between them is bounded by some threshold $\radius{\anypoint}$ that we will refer to as a radius. 
For now, it suffices to know that this value depends on the arrival location $\anypoint$ of $\request$ and will be defined later. 
To introduce formally the radius, we use the following notation for closed and open balls.

\begin{definition} \label{definition:balls}
    For each point $\anypoint \in \setofpoints$, let $\overline{B}(\anypoint,u)$ (resp.\ ${B}^\circ(\anypoint,u)$) denote the set of metric points $y \in \setofpoints$ with a distance no more than (resp.\ strictly less than) $u$ from $\anypoint$.
\end{definition}

The next part of the analysis heavily depends on whether there exists a request arriving after $r$ at any point in $\overline{B}(\anypoint,\radius{\anypoint})$ or not. 
When the latter happens, we call $r$ a late request and upper bound the cost of serving it by the highest value possible --- the sum of the metric space diameter and the expected waiting time for the next request to arrive. 
Although the estimation may seem exaggerated, it can be proved that only a few such requests exist. 
For the first case, when a close request arrives after $r$, with the right choice of $\radius{\anypoint}$, the expected cost of serving $\request$ can be upper bounded by a constant times the radius. 
We define radius as follows. For any subset of points $\mathcal{S}\subseteq \setofpoints$, we denote $\lambda(\mathcal{S}) := \sum_{\anypoint \in \mathcal{S}} \lambda_x$. 

\begin{definition} \label{definition:radius}
    For each point $\anypoint \in \setofpoints$, define its radius $\radius{\anypoint}$ as the minimum value $u \ge 0$, such that 
    \begin{equation*}
        \frac{1}{\lambda(\overline{B}(\anypoint, u))} \le u.
    \end{equation*}
\end{definition}
The idea behind it is that it balances the relationship between the diameter of $\overline{B}(\anypoint, u)$ and the expected waiting time between consecutive request arrivals within the points of this ball. 
Indeed, using the information from the preliminaries, one can show that the latter is equal to the left-hand side of the inequality. 
Finally, we note that the radius is well-defined as the function $u \mapsto 1 / \lambda(\overline{B}(x,u))$ is non-increasing and thus $\radius{\anypoint} \in (0, 1 / \lambda_{\anypoint}]$. 
See Figure \ref{figure:definingradius} for a pictorial example.

\begin{figure}[ht]
\input{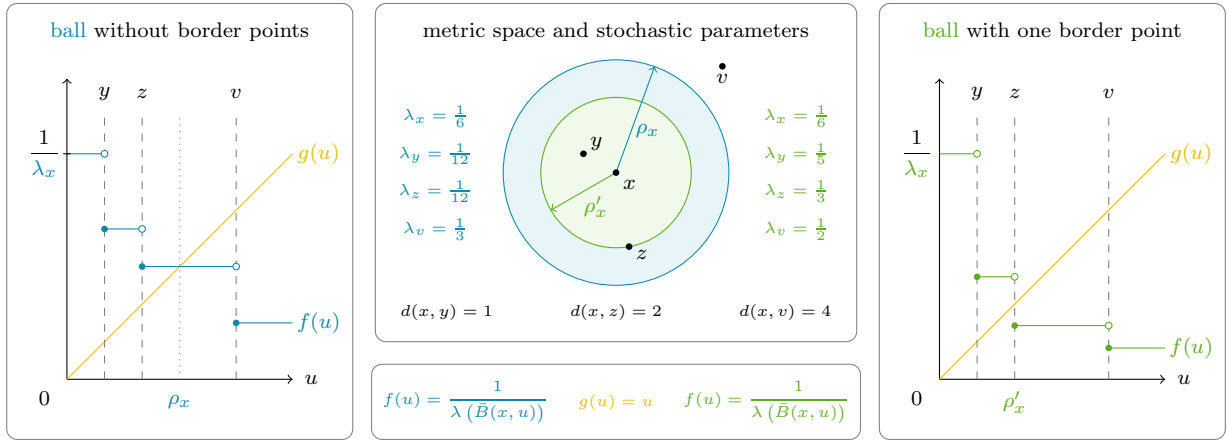}
\vspace{-16pt}
\caption{When determining the radius for some point $\anypoint$, two cases may occur. First, we have an equality in Definition \ref{definition:radius}, meaning that the plots of $f(u) = 1 / \lambda\left(\overline{B}(x,u)\right)$ and $g(u) = u$ intersect explicitly (see the graph on the left). Second, the value of $f(u)$ drops below $g(u)$ when approaching some point on the border of the ball (see the graph on the right).}
\label{figure:definingradius}
\end{figure}

By the radius definition, we have the following observation.
\begin{observation} \label{observation:lowerboundingscheme}
    Given any point $\anypoint \in \setofpoints$,
    \begin{equation*}
        \frac{1}{\lambda({B}^\circ(\anypoint, \radius{\anypoint}))}  \ge \radius{\anypoint} \ge \frac{1}{\lambda(\overline{B}(\anypoint, \radius{\anypoint}))}.
    \end{equation*}
\end{observation}

\noindent Here, we present both the upper and the lower bound on the radius, as one of them is needed to lower bound the expected cost of the optimal offline solution, and the second one is required to upper bound the expected cost of our  algorithms.

To conclude, let us state the upper bound on the expected cost produced by Greedy.

\begin{restatable}{lemma}{lemmagreedy}
\label{lemma:greedy_upperbound}
For MPMD in the Poisson arrival model, the expected cost produced by the Greedy algorithm, over all random sequences consisting of $m$ requests, satisfies
\begin{equation*}
    \E_{\sigma}^m[\greedy(\sigma)] \le \left( 4m \sum_{\anypoint \in \setofpoints} \frac{\lambda_x}{\lambda(\setofpoints)} \cdot \radius{\anypoint}\right)+2|\setofpoints| \cdot \left(d_{\max}+\frac{1}{\lambda(\setofpoints)}\right).
\end{equation*}
where $d_{\max} := \displaystyle\max_{x,y \in \setofpoints} d(x,y)$ is the diameter of the metric space. 
\end{restatable}

\noindent The last term of the right-hand side describes the cost of serving the late requests. 
The first term represents the standard expected cost of serving requests and is proportional to the length of the sequence $\sigma$. 
We prove this lemma in Section \ref{section:proof_greedy}. 

\subsection{The Radius algorithm}
In this subsection, our goal is to improve the performance guarantees of the Greedy algorithm on stochastic inputs.  
For this purpose, we design a Radius algorithm that calculates the radii upfront and uses this information to serve the requests better. 
The main idea here is to match any two requests whenever the closed balls of their locations (with radii defined as in Definition \ref{definition:radius}) overlap.

In the geometrical interpretation, whenever a request arrives at some point $\anypoint$, the algorithm directly sets its ball to be $\overline{B}(\anypoint,\radius{\anypoint})$. 
Hence, once a request $\request$ appears, if its location belongs to the closed ball of any pending request $\request'$, then the two are matched\footnote{Notice that there exists at most one such request $\request'$. Otherwise, if at the moment of its arrival, $\request$ belonged to the closed balls of two requests $\request'$ and $\request''$, their balls would intersect. Thus, they should have been paired before, which leads to a contradiction.}.  
Otherwise, if there exists another request $\request''$ within the distance of $\radius{\location{\request}} + \radius{\location{\request''}}$ from $\request$'s location, $\request$ can be matched with any such request. 
Finally, if no request satisfies the above conditions, $\request$ is temporarily left unmatched. 
See the pseudo-code shown in Algorithm \ref{pseudocode-radius-original} for a precise description of Radius. 
Notice that since Radius calculates the radii, it needs to know the metric space $(\setofpoints,d)$ and the exponential parameters $\{\lambda_x\}_{x \in \setofpoints}$. 
This is not an heavy requirement, since in the case of stochastic inputs, by the Law of large numbers, one can learn in constant time $O(1/\min_{x\in\setofpoints}\lambda_x)$ an arbitrarily good estimate of the arrival rates. 

\RestyleAlgo{boxruled}
\LinesNumbered
\SetAlgoVlined
\begin{algorithm}
\caption{Radius}
\label{pseudocode-radius-original}
\KwIn{A sequence $\sigma=(r_1,\dots, r_m)$ of requests, the arrival rate of each metric point.}
\KwOut{A perfect matching of the requests.}
Compute the radius $\radius{\anypoint}$ for each point $\anypoint \in \setofpoints$\ (Definition \ref{definition:radius})\;
$P \gets$ the set of pending requests, initially empty\;
\For{$i=1$ to $m$}{
    let $t=t(r_i)$ denote the arrival time of the $i$-th request $r_i$\;
    \If{there exists a pending request $r'\in P$ such that $d(\ell(r_i), \ell(r'))\le \rho_{\ell(r')}$}{
    match $r_i$ and $r'$ together, and remove $r'$ from $P$.
    }
    \ElseIf{there exists a pending request $r'\in P$ such that $d(\ell(r_i), \ell(r'))\le \rho_{\ell(r')}+\rho_{\ell(r_i)}$}{
    match $r_i$ and $r'$ together, breaking ties arbitrarily, and remove $r'$ from $P$.
    }
    \Else{
    add $r_i$ in $P$
    }
}
    \If{$P\neq \emptyset$}{
    match all requests in $P$ arbitrarily
    }
\end{algorithm}

It turns out that using the radius information directly in the algorithm leads to a better ratio of expectations. 
Below we present an upper bound on the expected cost of the Radius solution.

\begin{restatable}{lemma}{lemmaradius} \label{lemma:radius_upperbound}
    For MPMD in the Poisson arrival model, the expected cost produced by the Radius algorithm, over all random sequences consisting of $m$ requests, satisfies
    \begin{equation*}
        \E_{\sigma}^m[\rad(\sigma)] \le \left(2m \sum_{\anypoint \in \setofpoints} \frac{\lambda_x}{\lambda(\setofpoints)} \cdot \radius{\anypoint}\right) + \frac{1}{2} \cdot |\setofpoints| \cdot d_{\max}.
    \end{equation*}
    where $d_{\max} := \displaystyle\max_{x,y \in \setofpoints} d(x,y)$ is the diameter of the metric space. 
\end{restatable}

The formal proof of this lemma can be found in Section \ref{section:proof_radius}. 
Here, we conclude the algorithm description with an example illustrated in Figure \ref{figure:radius_algorithm}.

\begin{figure}[ht]
\input{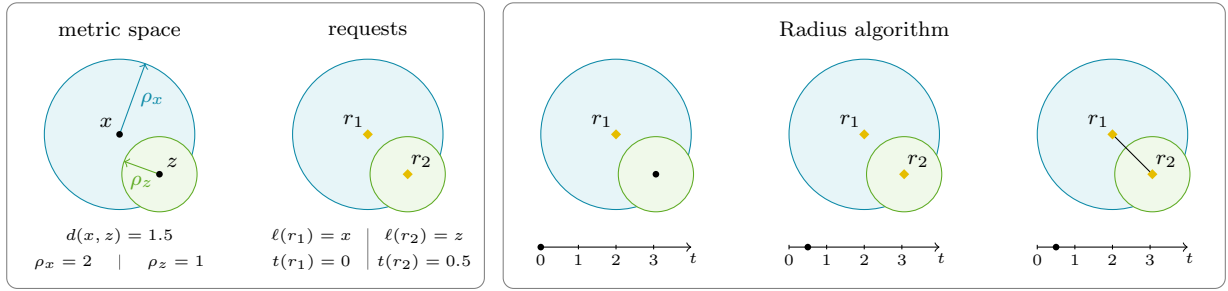}
\vspace{-16pt}
\caption{An example of how Radius works on a sequence of two requests.}
\label{figure:radius_algorithm}
\end{figure}

\subsection{Lower bound on the optimal offline matching}
It remains to present an overview of the lower bounding scheme for the optimal offline solution of the MPMD problem. 
Having such a result will enable us to find the performance ratio for the two algorithms introduced before and show that they both achieve constant ratio of expectations.

The crucial part of the lower bounding process is to analyze each request $\request$ in a sequence $\sigma$ separately and observe that two situations can happen when $\request$ is not matched immediately. 
On one hand, $\request$ can be matched early with some distant request, thus, paying a high connection cost. 
On the other hand, it can wait for a closer request to arrive and pay a higher delay cost. 
A similar situation takes place when $\request$ is paired at the moment of its arrival with an older request. 
The only difference then is that we go through the timeline in the opposite direction.

Let us set the threshold for a request to be considered close to $\request$ as the radius of $\request$'s arrival location, i.e., $\radius{\location{\request}}$. 
Then, the expected cost of serving $\request$ can be upper bounded by the expected value of the minimum of three things. 
The first one is the cost of matching $\request$ with the latest request that has arrived in ${B}^\circ(\anypoint,\radius{\anypoint})$ before $\request$. 
The second is equal to the cost of matching $\request$ with the earliest request arriving after it at any point in this ball. 
Finally, the third one is just the radius $\radius{\anypoint}$ as it is the lower bound for the connection cost outside the ball. 
When we use the stochastic assumption to compute this minimum, we obtain the following.
\begin{restatable}{lemma}{lowerbound} \label{lemma:lowerboundingscheme}
    For MPMD in the Poisson arrival model, the expected cost of the optimal offline solution, over all random sequences consisting of $m$ requests, satisfies
    \begin{equation*}
        \E_\sigma^m[\optim(\sigma)] \ge m \cdot \frac{1-e^{-2}}{4} \sum_{\anypoint \in \setofpoints}\frac{\lambda_x}{\lambda(\setofpoints)} \cdot \rho_x.
    \end{equation*}
\end{restatable}

\noindent We present a detailed proof of this lemma in Section \ref{section:lowerbound}.

\subsection{Proofs of the main theorems}
Finally, we prove the two main theorems stated in the introduction.

\greedyapprox*
\begin{proof}
Using the upper bound on Greedy from Lemma \ref{lemma:greedy_upperbound} and the lower bound on the optimal offline solution presented in Lemma \ref{lemma:lowerboundingscheme}, we obtain
\begin{equation*}
\begin{split}
    \overline{\lim_{m \to \infty}} \dfrac{\E_{\sigma}^m[\algor(\sigma)]}{\E_{\sigma}^m[\optim(\sigma)]} &\ \leq\ \lim_{m \to \infty} \dfrac{\left( 4m \sum_{\anypoint \in \setofpoints} \dfrac{\lambda_x}{\lambda(\setofpoints)} \cdot \radius{\anypoint}\right)+2|\setofpoints| \cdot \left(d_{\max}+\dfrac{1}{\lambda(\setofpoints)}\right)}{m \cdot \dfrac{1-e^{-2}}{4} \sum_{\anypoint \in \setofpoints}\dfrac{\lambda_x}{\lambda(\setofpoints)} \cdot \rho_x} \\[4pt]
    &\ =\ \dfrac{16}{1-e^{-2}} + \lim_{m \to \infty} \dfrac{1}{m} \cdot \dfrac{2|\setofpoints| \cdot \left(d_{\max}+\dfrac{1}{\lambda(\setofpoints)}\right)}{\dfrac{1-e^{-2}}{4} \sum_{\anypoint \in \setofpoints}\dfrac{\lambda_x}{\lambda(\setofpoints)} \cdot \rho_x}\ =\ \dfrac{16}{1-e^{-2}},
\end{split}
\end{equation*}
which concludes the proof.\footnote{Note that the assumption of finite points in the given metric, i.e. $|X| < \infty$, is necessary to prove the theorem.}
\end{proof}

\noindent Analogously, if we refer to Lemma \ref{lemma:radius_upperbound}, we can prove the following.

\radiusapprox*

\section{Lower bound on the optimal offline matching}
\label{section:lowerbound}
In this section, we prove Lemma \ref{lemma:lowerboundingscheme} (restated below). 

\lowerbound*

\noindent The main idea of the proof goes as follows. 
To obtain a lower bound on the expected cost of the optimal matching over a request sequence $\sigma$, we analyze each element of $\sigma$ separately. 
First, we observe that for each request, the sum of its connection and delay cost is at least equal to the cost of connecting it to its cheapest neighbor in $\sigma$ (in terms of connection + delay cost). 
Then, the core of the proof (Claim \ref{claim:expected_cost_x}) consists of showing that, in expectation, this cost is at least a constant times the radius of the corresponding point. \vspace{6pt}

\noindent Given any input sequence $\sigma$ and any request $\request \in \sigma$, we define the \emph{minimum total cost of $r$ in $\sigma$} as
\begin{equation*}
    c(\sigma, r) := \min_{r' \in \sigma, r' \neq r} \left\{d(\ell(r),\ell(r')) + |t(r) - t(r')|\right\}.
\end{equation*}

\begin{claim} \label{claim:opt_tilde}
    For any input sequence $\sigma$ it holds that $\optim(\sigma) \ge \frac{1}{2}\sum_{\request \in \sigma} c(\sigma,r)$. 
\end{claim}

\begin{proof}
Given any pair $(r,r')$ that gets matched by $\optim$, its connection cost is $d(\ell(r),\ell(r'))$ and its delay cost equals $|t(r)-t(r')|$.
By definition, we have $c(\sigma, r) \le d(\ell(r),\ell(r')) + |t(r)-t(r')|$ as well as $c(\sigma,r')\le d(\ell(r),\ell(r')) + |t(r)-t(r')|$, which gives $d(\ell(r),\ell(r')) + |t(r)-t(r')| \ge \frac{c(\sigma,r) + c(\sigma,r')}{2}$. 
Finally, we obtain the claim by summing over all matched pairs in $\optim(\sigma)$. 
\end{proof}

Before formally stating Claim \ref{claim:expected_cost_x}, we need the following two results. 

\begin{restatable}{proposition}{subsetdistribution} \label{proposition:subset_exponential_waiting_time}
    Let $\sigma = (r_1,r_2,\ldots)$ be an infinite sequence of requests generated by the centralized Poisson process and ordered by their arrival times. 
    Then, for any point $x\in\setofpoints$ and any index $i \geq 1$, the distribution of the waiting time for a next request to arrive after $r_i$ at some point in a set $S \subseteq \setofpoints$, $x \in S$, follows an exponential distribution with parameter $\lambda(S)$.
\end{restatable}

\begin{proof}
By Definition \ref{definition:centralized_poisson_arrival} we know that the waiting times between the consecutive requests arrivals in $\sigma$ follow an exponential distribution with parameter $\lambda(\setofpoints)$ and are independent from the requests location choices. 
They are also mutually independent by the arrival model's definition.

Notice that, if we look at any point $y \in \setofpoints$, the probability that the $k$-th request, $k \ge 1$, arrives at $y$ is $\lambda_y / \lambda(\setofpoints)$. 
Thus, from the perspective of $y$, the exact distribution of the arrival rates $\lambda$ of the other points does not matter --- it only needs to know the value of $\lambda(\setofpoints\setminus\{y\})$ to compute its chances of being chosen. 
This means that to find the waiting time for a next request to arrive at any point in $\S$, we can merge all the points in $\S$ into one new point $s$ with the parameter $\lambda(\S)$.

What remains is to recall that the centralized Poisson process was equivalent to the model where we used timers and reset them every time a request arrived. However, we could also decide not to reset a timer. Doing so for the $s$'s timer implies that the waiting time between $t(r_i)$ and the arrival of the next request at $s$ has an exponential distribution with parameter $\lambda(\S)$.
\end{proof}

\noindent To get a better understanding of the above analysis, see Figure \ref{figure:subset_distribution}. We also present an elementary proof of this proposition, based only on the centralized Poisson model, in Appendix \ref{app:subset_distribution_elementary}.

\begin{figure}[ht]
\input{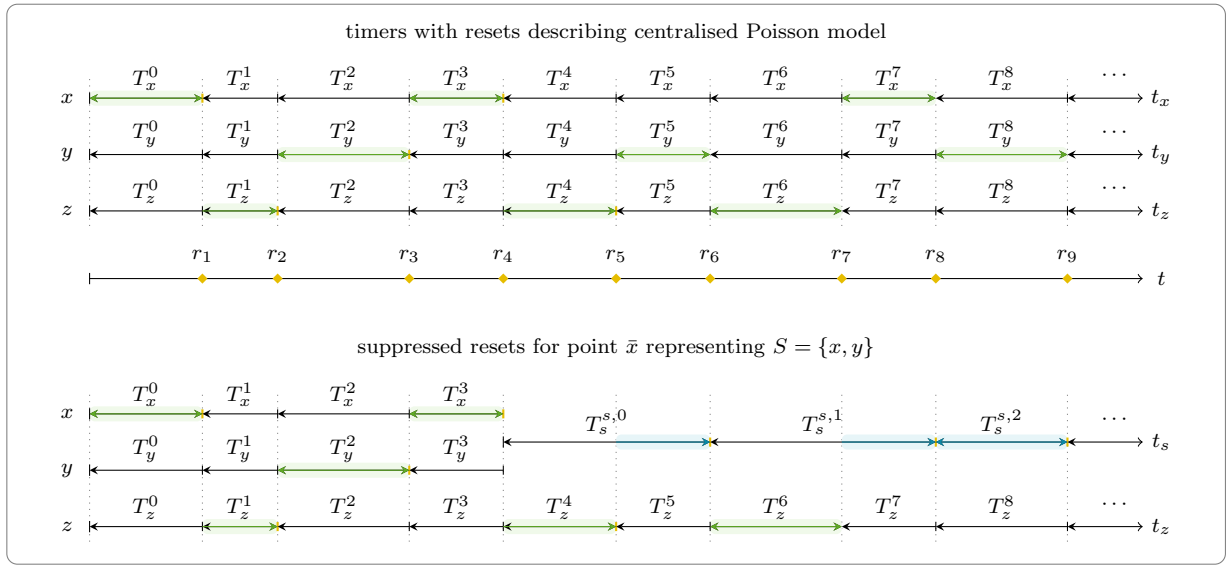}
\vspace{-16pt}
\caption{An example showing the main idea presented in the proof of Proposition \ref{proposition:subset_exponential_waiting_time}. Here, to analyze the waiting time between consecutive arrivals on points $x$, $y$, we define new timers $T_s^{s,i}$ for the sum of these point.}
\label{figure:subset_distribution}
\end{figure}

\begin{claim} \label{claim:exponential-2}
    Given any $a > 0$ and an exponential variable $Y \sim \expdistr{\mu}$, $\E[\min\{Y, a\}] = \frac{1 - e^{-\mu a}}{\mu a} \cdot a$.
\end{claim}

\begin{proof}
We simply calculate the expected value obtaining
\begin{equation*}
    \E[\min\{Y, a\}]
    = \int_0^a t \cdot \mu e^{-\mu t} dt + \int_a^\infty a\cdot \mu e^{-\mu t} dt
    = \left[-(t+\frac{1}{\mu}) e^{-\mu t}\right]_0^a + \left[-ae^{-\mu t}\right]_a^\infty
    = \frac{1 - e^{-\mu a}}{\mu a} \cdot a,
\end{equation*}
which ends the proof.
\end{proof}

Now, let us present the core component needed to prove Lemma \ref{lemma:lowerboundingscheme}.

\begin{claim} \label{claim:expected_cost_x}
    Given a sequence $\sigma$, we order the requests in $\sigma = (r_1, \dots, r_m)$ according to their arrival times. 
    Then, for any point $x \in \setofpoints$ and any index $i \in \{1, \dots, m\}$, the expected minimum cost of the $i$-th request $r_i$ in a random sequence $\sigma$, assuming that $r_i$ is located on $x$, is
\begin{equation*}
    \E_\sigma^m[c(\sigma, r_i) \mid \ell(r_i) = \anypoint] \ge \frac{1-{e}^{-2}}{2} \cdot \radius{\anypoint}.
\end{equation*}
\end{claim}

\begin{proof}
To facilitate the proof, we first extend every random request sequence $\sigma=(r_1, \dots, r_m)$ by some dummy random requests $r_j$ for $j \le 0$ and $j \ge m$ to get an \emph{extended} random sequence
\begin{equation*}
    \overline{\sigma} = (\dots, r_{-2},r_{-1},r_0,r_1,\dots, r_{m-1},r_{m},r_{m+1},\dots).
\end{equation*}
To generate requests before $r_1$ and after $r_m$ we use the centralized Poisson arrival model (i.e., for every integer $j$, $(t(r_{j+1}) - t(r_j)) \sim \expdistr{\lambda(\setofpoints)}$ and $\prob(\ell(r_j) = y) = \lambda_y / \lambda(\setofpoints)$ for all $y \in \setofpoints$). 
This implies that in the extended sequence of requests, given any point $y \in \setofpoints$, with probability one there exist indexes $j \le 0$ and $j' \ge m + 1$ satisfying $\ell(r_j) = \ell(r_{j'}) = y$.

For an extended random sequence $\overline{\sigma}$, we define its truncation $\overline{\sigma}_m := (r_1,\dots, r_m)$. 
Remark that the probability distribution of the truncated extended random request is identical to the original distribution of the random sequences of $m$ requests. Moreover, the minimum cost $c(\sigma, r_j)$ of any request $r_j$, with $j \in \{1, \dots, m\}$, can only decrease, i.e., $c(\overline{\sigma}, r_j) \le c({\sigma}, r_j)$ where $\sigma = \overline{\sigma}_m$.
Hence, we have
\begin{equation*}
    \E_\sigma^m[c(\sigma, r_i) \mid \ell(r_i) = \anypoint] \ge \E_{\overline{\sigma}}[c(\overline{\sigma}, r_i) \mid \ell(r_i) = \anypoint].
\end{equation*}
Notice that the conditional expected minimum cost of each $r_j$ arriving at any point $\anypoint$ in an extended random sequence is now the same for all request of the sequence, i.e., for any $j, j'\in \{0, 1, \dots, m\}$,  
\begin{equation*}
    \E_{\overline{\sigma}}[c(\overline{\sigma}, r_j) \mid \ell(r_j) = \anypoint] = \E_{\overline{\sigma}}[c(\overline{\sigma}, r_{j'}) \mid \ell(r_{j'}) = \anypoint].
\end{equation*}
In particular, to prove the claim, we only need to establish a bound on the conditional expected cost of serving the request zero, i.e., we need to show that
\begin{equation} \label{eq:to_prove}
    \E_{\overline{\sigma}}[c(\overline{\sigma}, r_0) \mid \ell(r_0) = \anypoint] \ge \frac{1-{e}^{-2}}{2} \cdot \radius{\anypoint}.
\end{equation}
To prove this bound, consider an extended sequence $\overline{\sigma}$ with $\ell(r_0) = \anypoint$. 
W.l.o.g., we also assume that $t(r_0) = 0$ as it can be achieved by shifting all arrival times by the same constant. 
Define $W^-$ (resp.\ $W^+$) as the (random) time duration between the arrival of the last request before $r_0$ (resp.\ first request after $r_0$) arriving at any point $y \in B^\circ(x,\rho_x)$ and the arrival of $r_0$.
Formally,
\begin{equation*}
    W^- := \min_{j < 0}\left\{-t(r_j): d(\ell(r_j),\anypoint) < \radius{\anypoint} \right\} \ge 0 \quad \text{ and } \quad W^+ := \min_{j > 0}\left\{t(r_j): d(\ell(r_j),x) < \radius{\anypoint}\right\} \ge 0.
\end{equation*}
Since we work with extended sequences, both $W^+$ and $W^-$ are finite. Hence, we can use them to lower bound the value of $c(\overline{\sigma},r_0) < c(\sigma, r_0)$.\footnote{Recalling from the definition of function $c(\sigma, r)$ (which takes minimum value over all the requests in $\sigma$) as well as $\sigma \subseteq \overline{\sigma}$, we can conclude that $c(\overline{\sigma}, r_0)$ can only become smaller than $c(\sigma, r_0)$.} 
Intuitively, we lower bound the (connection + delay) cost of the requests far from $x$ by their connection cost, and the total cost of the requests close to $x$ by their delay cost. More precisely, 
\begin{eqnarray} \label{lower-bound-cost-request}
    c(\overline{\sigma},r_0)
    &=&\min_{j\neq 0} \big\{d(x,\ell(r_j)) + |t(r_j)|\big\} \nonumber\\
    &\ge& \min\Big\{\min_{j\neq 0}\left\{|t(r_j)|:  d(\ell(r_j), \anypoint) < \rho_x \right\}, \radius{\anypoint}\Big\} \nonumber \\
    &=& \min\Big\{\min_{j < 0} \left\{-t(r_j): d(\ell(r_j), \anypoint) < \rho_x\right\},\min_{j> 0} \left\{t(r_j): d(\ell(r_j), \anypoint) < \rho_x\right\}, \radius{\anypoint} \Big\} \nonumber \\
    &=& \min\big\{W^-, W^+, \radius{\anypoint}\big\} \nonumber \\
    &=& \min\big\{\min\{W^-, W^+\}, \radius{\anypoint}\big\},
\end{eqnarray}
where the first inequality is obtained by lower bounding $d(x,\ell(r_j)) + |t(r_j)|$ by $\rho_x$ when $d(x,\ell(r_j))\ge \rho_x$, and by $|t(r_j)|$ otherwise. 

We claim that $W^-$ and $W^+$ are mutually independent and follow the same exponential distribution $\expdistr{\lambda(B^{\circ}(\anypoint, \radius{\anypoint}))}$. 
To prove it, let us recall that we work with the centralized Poisson arrival model. 
Thus, we can say that $W^+$ (resp.\ $W^-$) depends on the waiting times and arrival location choices after (resp.\ before) $t = 0$. 
Since the variables representing history before $t = 0$ in our model are independent from those representing history after $t = 0$, so are $W^-$ and $W^+$. 
It is also easy to prove that $W^+$ has an exponential distribution with parameter $\lambda(B^{\circ}(\anypoint, \radius{\anypoint}))$ --- it follows straightforward from Proposition \ref{proposition:subset_exponential_waiting_time}. 
However, if we look at the timeline for $\overline{\sigma}$ and go from $t = 0$ to $-\infty$, we can once again use Proposition \ref{proposition:subset_exponential_waiting_time} for the sequence $(r_0, r_{-1}, \dots)$.
In this way, we prove that the distribution of $W^-$ is also exponential with parameter $\lambda(B^{\circ}(\anypoint, \radius{\anypoint}))$.

Now, thanks to Proposition \ref{proposition:minimum}, we immediately have
\begin{equation*}
    \min\{W^-, W^+\} \sim \expdistr{2\lambda(B^\circ(x,\rho_x))}. 
\end{equation*}

\noindent By inequality \eqref{lower-bound-cost-request} and Claim \ref{claim:exponential-2} (with $a = \radius{\anypoint}$, $\mu = 2\lambda(B^\circ(\anypoint, \radius{\anypoint}))$), it follows that
\begin{equation*}
    \E_{\overline{\sigma}}[c(\overline{\sigma}, r_0) \mid \ell(r_0) = \anypoint] 
    \ge \E_{\overline{\sigma}}\Big[\min\big\{\min\{W^-, W^+\}, \radius{\anypoint}\big\}\Big] = \frac{1 - e^{- 2\lambda(B^\circ(\anypoint, \radius{\anypoint})) \cdot \radius{\anypoint}}}{2\lambda(B^\circ(\anypoint, \radius{\anypoint})) \cdot \radius{\anypoint}} \cdot \radius{\anypoint}.
\end{equation*}
It is easy to check that function $t\mapsto\frac{1 - e^{-t}}{t}$ is strictly decreasing. 
Together with Observation \ref{observation:lowerboundingscheme} it guarantees that $\lambda(B^\circ(\anypoint, \radius{\anypoint})) \cdot \radius{\anypoint} \le 1$ and thus 
\begin{equation*}
    \frac{1 - e^{- 2\lambda(B^\circ(\anypoint, \radius{\anypoint})) \cdot \radius{\anypoint}}}{2\lambda(B^\circ(\anypoint, \radius{\anypoint})) \cdot \radius{\anypoint}} \ge \frac{1 - e^{-2}}{2}.
\end{equation*}
Hence,
\begin{equation*}
    \E_{\overline{\sigma}}[c(\overline{\sigma}, r_0) \mid \ell(r_0) = \anypoint] \ge \frac{1 - e^{-2}}{2} \cdot \radius{\anypoint},
\end{equation*}
which concludes the proof of inequality \eqref{eq:to_prove} and, what follows, the proof of Claim \ref{claim:expected_cost_x}. 
\end{proof}

Finally, we prove Lemma \ref{lemma:lowerboundingscheme}. 
\begin{proof}[Proof of Lemma \ref{lemma:lowerboundingscheme}]
Let $\sigma = (r_1, \dots, r_m)$ be a sequence of request sorted in an increasing order of their arrival times. We have
\begin{align*}
    \E_\sigma^m[\optim(\sigma)] &\ge \E_\sigma^m\left[\frac{1}{2}\sum_{i=1}^m c(\sigma,r_i)\right] &&\text{(Claim \ref{claim:opt_tilde})}\\
    &= \frac{1}{2} \sum_{i=1}^m \E_\sigma^m[c(\sigma,r_i)] &&\text{(linearity of expectation)}\\
    &= \frac{1}{2} \sum_{i=1}^m \sum_{x\in \setofpoints}\prob_\sigma(\ell(r_i)=x)\cdot\E_\sigma^m[c(\sigma,r_i)\mid \ell(r_i)=x]\\
    &\ge \frac{1}{2} \sum_{i=1}^m \sum_{x\in \setofpoints} \frac{\lambda_x}{\lambda(\setofpoints)}\cdot\frac{1-{e}^{-2}}{2} \cdot \radius{\anypoint} &&\text{(Claim \ref{claim:expected_cost_x})}\\
    &= m \cdot \frac{1-e^{-2}}{4} \sum_{\anypoint \in \setofpoints}\frac{\lambda_x}{\lambda(\setofpoints)} \cdot \rho_x.
\end{align*}
This concludes the proof.
\end{proof}

\section{Upper bound on the Greedy solution}
\label{section:proof_greedy}

In this section, we prove Lemma \ref{lemma:greedy_upperbound} (restated below) that establishes an upper bound on the expected cost of the Greedy algorithm for stochastic inputs. 

\lemmagreedy*

\noindent To prove this upper bound, we first observe that the total connection cost of the Greedy solution is at most equal to its total delay cost, and then we bound the latter. 

Given any input sequence $\sigma$, let $(M,s)$ denote the solution output by the Greedy algorithm, where $M$ is the set of matched pairs of request, and $s$ is the service times of the requests. 
The waiting time of a request $r \in \sigma$ is denoted by $w(r) := s(r) - t(r)$.  
Greedy matches two requests $r$ and $r'$ when the sum of their delay cost $w(r) + w(r')$ is at least equal to their distance  $d(\ell(r), \ell(r'))$. 
In particular, when summing over all requests we obtain:
\begin{claim}
\label{claim:greedy_connected_at_most_delay}
For any input sequence $\sigma$, the cost of the solution returned by the Greedy algorithm is at most twice its total delay cost, i.e., 
\begin{equation*}
    \greedy(\sigma) \le 2\sum_{r \in \sigma} w(r).
\end{equation*}
\end{claim}

We now focus on bounding the waiting time of each request. 
To do this, we distinguish two types of requests. 
For each request $r$, define $t'(r) := t(r) + \rho_{\ell(r)}$.
We say that $r$ is a \emph{late} request if 
\begin{itemize}
    \item[-] $r$ is still pending at time $t'(r)$ and
    \item[-] there is no request $r'$ arriving within the closed ball of $r$'s location (i.e., $d(\ell(r), \ell(r')) \le \rho_{\ell(r)}$) after time $t'(r)$.
\end{itemize}
Otherwise, we say that $r$ is a \emph{nice} request, and we define
\begin{equation*}
    Y^\text{nice}_r :=
    \begin{cases}
    0 & \textit{ if } r \textit{ is matched at time } t'(r);\\
    \displaystyle\min_{r'\in\sigma} \left\{t(r') - t'(r) \mid t(r') > t'(r) \text{ and } d(\ell(r'),\ell(r)) \le \rho_{\ell(r)}\right\} & \textit{ otherwise.}
    \end{cases}
\end{equation*}
We bound the waiting time of nice requests as follows:
\begin{claim}
    For each nice request $r \in \sigma$, we have $w(r) \le \rho_{\ell(r)} + Y^\text{nice}_r$.
\end{claim}

\begin{proof}
Let $r$ be any nice request in $\sigma$. 
If $r$ has already been matched at time $t'(r)$ then we have $w(r) \le \rho_{\ell(r)}$.  
Otherwise, let $r' \in \sigma$ be the first request satisfying $t(r') > t'(r)$ and $d(\ell(r'),\ell(r)) \le \rho_{\ell(r)}$, i.e.,  $t(r') - t(r) = \rho_{\ell(r)} + Y^\text{nice}_r$. 
By definition of nice request, such a request $r'$ exists.   

In the case that $r$ has already been matched at time $t'(r)$, we have $w(r) \le \rho_{\ell(r)}$;
otherwise, we claim that $r$ and $r'$ are matched together by the Greedy algorithm at time $t = t(r')$. 
Indeed, we have 
\begin{equation*}
    (t - t(r)) + (t - t(r')) = \rho_{\ell(r)} + Y^\text{nice}_r > \rho_{\ell(r)} \ge d(\ell(r'),\ell(r)),
\end{equation*} 
so the greedy criteria is satisfied by the pair $(r,r')$. 
Suppose for a contradiction that $r'$ is matched at time $t$ with another pending request $r''$. 
It is necessary because $t - t(r'') \ge d(\ell(r'),\ell(r''))$. 
With the triangle inequality, we obtain
\begin{eqnarray}
    (t - t(r)) + (t - t(r'')) &>& d(\ell(r),\ell(r')) + d(\ell(r'),\ell(r'')) \nonumber\\
    &\ge& d(\ell(r),\ell(r'')). \nonumber
\end{eqnarray}
This means that $r$ and $r''$ should have been matched together before the arrival of $r'$, which leads to a contradiction.
\end{proof}

We now bound the total delay time induced by late requests. 
Unfortunately, the waiting time of a late request can possibly be as large as the diameter $d_{\max} = \max_{x, y \in \setofpoints} d(x,y)$ of the metric space. 
However, we show that there are only constantly many such requests. 
Let $t(r_m)$ denote the arrival time of the last request in $\sigma$. 
For any late request $r$, define 
\begin{equation*}
    Y^\text{late}_r :=
    \begin{cases}
    0 & \textit{ if } t(r) + d_{\max} \ge t(r_m);\\
    \displaystyle\min_{r' \in \sigma} \left\{t(r') - (t(r) + d_{\max}) \mid t(r') > t(r) + d_{\max}\right\} & \textit{ otherwise.}
    \end{cases}
\end{equation*}

\begin{claim} \label{claim:few_late_requests_greedy}
    For any point $x\in \setofpoints$, there is at most one late request located on $x$. 
    In particular, there are at most $|\setofpoints|$ late requests. 
    Moreover, for each late request $r$, we have $w(r) \le d_{\max} + Y^\text{late}_r$. 
\end{claim}

\begin{proof}
Let $r$ and $r'$ be the two requests such that $\ell(r)=\ell(r')$ and $t(r) < t(r')$. 
Suppose for the sake of a contradiction that both $r$ and $r'$ are late. 
By definition of late request, this implies that $t(r') \le t(r) + \rho_{\ell(r)}$, and in particular, $r$ is still pending at time $t(r')$. 
Since $r'$ is late, it is not matched by the Greedy algorithm at its arrival, and in particular, $r$ and $r'$ are not matched together by Greedy. 
This is a contradiction with the greedy criteria since $d(\ell(r),\ell(r')) = 0 \le \rho_{\ell(r)} + \rho_{\ell(r')}$ (the Greedy algorithm should match $r$ and $r'$ together). 

We now show the second part of the statement. 
Let $r$ be any late request. 
If $t(r) + d_{\max} \ge t(r_m)$, then at time $t = t(r) + d_{\max}$, all the requests in $\sigma$ already arrived. 
Thus, either $r$ has already been matched, or there exists at least one other pending request $r'$ (as we assume that the total number of requests is even). 
Since $t - t(r) = d_{\max} \ge d(\ell(r),\ell(r'))$, the Greedy algorithm matches $r$ and $r'$ at time $t$, which implies that $w(r)=d_{\max}$. 
Otherwise, we have $t(r) + d_{\max} < t(r_m)$. 
By the definition of $Y^\text{late}_r$, the next request $r'$ in $\setofpoints$ arrives at time $t = t(r) + d_{\max} + Y^\text{late}_r$. 
Once again, either $r$ has already been matched at time $t$ (and its waiting time is at most $d_{\max} + Y^\text{late}_r$), or $r$ is matched with $r'$ by the same argument. 
In any case, we have proved that $w(r) \le d_{\max} + Y^\text{late}_r$. 
\end{proof}

We now use stochastic assumptions to upper bound the expected cost of the solution. 

\begin{proof}[Proof of Lemma \ref{lemma:greedy_upperbound}]
Let $\sigma = (r_1, \dots, r_m)$ be a random sequence of $m$ requests, where requests are ordered with increasing arrival times. 

We first bound the expected delay cost induced by late requests. 
Suppose that the $i$-th request of the sequence, $r_i$, is late and is located on a point $x\in \setofpoints$. 
Using Definition \ref{definition:centralized_poisson_arrival}, we know that when $\sigma$ is a random sequence generated by the Poisson arrival process, the (conditional) random variable $Y^\text{late}_r$ follows an exponential distribution of parameter $\lambda(\setofpoints)$. 
In particular, we obtain
\begin{equation*}
    \E_\sigma^m[w(r_i) \mid r_i \text{ is late and }\ell(r_i) = x] \le d_{\max} + \frac{1}{\lambda(\setofpoints)}. 
\end{equation*}
Since the expectation does not depend on point $x$, and since there are at most $|\setofpoints|$ late requests, the total delay cost induced by the late requests is in expectation:
\begin{equation*}
    \E_\sigma^m\left[\sum_{\substack{i=1 \\r_i\text{ is late}}}^m w({r_i})\right]\le |\setofpoints| \cdot \left(d_{\max} + \frac{1}{\lambda(\setofpoints)}\right).
\end{equation*}
When $r_i$ is a nice request located on $x \in \setofpoints$, by Proposition \ref{proposition:subset_exponential_waiting_time}, the (conditional) random variable $Y^\text{nice}_r$ follows an exponential distribution of parameter $\sum_{y \in \setofpoints: d(x,y) \le \rho_x} \lambda_y$. 
Using Observation \ref{observation:lowerboundingscheme} we obtain:
\begin{equation*}
    \E_\sigma^m[w(r_i) \mid r_i \text{ is nice and } \ell(r_i) = x] = \rho_x + 1 / \lambda(\overline{B}(x,\rho_x)) \le \rho_x + \rho_x = 2\rho_x.
\end{equation*}
As we observed in Definition \ref{definition:centralized_poisson_arrival}, the probability the the $i$-th request of the (random) sequence is located on $x$ is equal to $\lambda_x / \lambda(\setofpoints)$. 
Thus, the total delay cost induced by nice requests is 
\begin{align*}
    \E_\sigma^m\left[\sum_{\substack{i=1 \\r_i\text{ is nice}}}^m w({r_i})\right]
    &\le \sum_{i=1}^m \sum_{x \in \setofpoints} \prob(r_i \text{ is nice and }\ell(r_i) = x) \cdot \E_\sigma^m\left[w({r_i}) \mid r_i \text{ is nice and }\ell(r_i) = x \right]\\
    &\le \sum_{i=1}^m \sum_{x \in \setofpoints} \prob(\ell(r_i)=x) \cdot 2\rho_x\\
    &\le \sum_{i=1}^m \sum_{x \in \setofpoints} \frac{\lambda_x}{\lambda(\setofpoints)} \cdot 2\rho_x\\
    &= m \sum_{x \in \setofpoints} \frac{\lambda_x}{\lambda(\setofpoints)} \cdot 2\rho_x.
\end{align*}
Finally, putting everything together, we obtain the expected bound:
\begin{align*}
    \E_\sigma^m\left[\greedy(\sigma)\right]
    &\le \E_\sigma^m\left[2\sum_{i=1}^m w({r_i})\right] &&\text{(Claim \ref{claim:greedy_connected_at_most_delay})}\\
    &= 2 \cdot \E_\sigma^m\left[\sum_{\substack{i=1 \\r_i\text{ is nice}}}^m w({r_i})\right] + 2 \cdot \E_\sigma^m\left[\sum_{\substack{i=1 \\r_i\text{ is late}}}^m w({r_i})\right] \\
    &\le \left(4m \sum_{x\in\setofpoints}\frac{\lambda_x}{\lambda(\setofpoints)}\cdot\rho_x \right) + 2 |\setofpoints| \cdot \left(d_{\max} + \frac{1}{\lambda(\setofpoints)}\right).
\end{align*}
This concludes the proof. 
\end{proof}

\section{Upper bound on the Radius solution}
\label{section:proof_radius}
In this section, we prove Lemma \ref{lemma:radius_upperbound} (restated below) that establishes an upper bound on the expected cost of the Radius algorithm.

\lemmaradius*

\noindent To bound the total cost of the solution produced by the Radius, we separately analyze the delay cost and the connection cost. 
To bound the connection cost we differentiate two types of edges (pairs matched by the algorithm).
Let $M$ denote the matching produced by the Radius algorithm on the input sequence $\sigma$. 
Let us call $e \in M$ a \emph{nice} edge\footnote{Notice that the current definition differs from the previous section.} if the corresponding matched pair was created during the main loop of the algorithm, i.e., before the arrival time of the last request in $\sigma$. 
Otherwise, we call this edge \emph{late}. 
Similarly, a request is called \emph{nice} if it is an endpoint of a nice edge, and \emph{late} otherwise.
Intuitively, since the late requests are matched arbitrarily by the Radius algorithm, the connection cost induced by each of these late edge can possibly be as large as the diameter of the metric space. 
Fortunately, we show that there are only a constant number of them (i.e., independent from $m$).
 
\begin{claim} \label{claim:few_late_edges}
    For any point $x \in \setofpoints$, there is at most one late request located on $x$. 
    In particular, there are at most $|\setofpoints| / 2$ late edges. 
\end{claim}

\begin{proof}
Let $r$ and $r'$ be two requests such that $\ell(r) = \ell(r')$ and $t(r) < t(r')$. 
Suppose that when the request $r'$ is processed by the radius algorithm, $r$ is still unmatched. 
Then, $0 = d(\ell(r),\ell(r')) \le \rho_{\ell(r)} + \rho_{\ell(r')}$, and the algorithm matches these requests together. 
This shows that $r$ cannot become a pending request, and in particular, it is not a late request. 
\end{proof}

We now bound the connection cost of the solution induced by nice edges. 
Two nice requests are matched together by the Radius algorithm if and only if their distance is at most the sum of their radii. 
By summing over all the nice edges, we obtain: 
\begin{claim} \label{claim:connection_cost_nice}
    For any input sequence $\sigma$, the connection cost induced by all nice edges is at most $\sum_{r \in \sigma} \rho_{\ell(r)}$. 
\end{claim}

We now bound the total delay cost. 
Let $t(r_m)$ denote the arrival time of the last request of $\sigma$, which correspond to the time at which all remaining pending (late) requests are matched together by the Radius algorithm. 
Let $r$ be any request in $\sigma$. 
We define $Y_r$ as the duration between the arrivals of $r$ and the first request $r'$ that appears on a point of $\overline{B}(\ell(r), \rho_{\ell(r)})$ after $r$. If there is no such request $r' \in \sigma$, then we set $Y_r := t(r_m) - t(r)$. 
Formally: 
\begin{equation*}
    Y_r := \min\left\{t(r_m) - t(r), \min_{r' \in \sigma}\left\{t(r') - t(r) \mid t(r') > t(r) \text{ and } d(\ell(r'), \ell(r))\le \rho_{\ell(r)}\right\}\right\}
\end{equation*}

\begin{claim} \label{claim:delay_cost_radius}
    Each request $r$ in $\sigma$ is delayed by the Radius algorithm for a time at most equal to $Y_r$. 
\end{claim}

\begin{proof}
Let $r$ be any request in $\sigma$, and assume that $r$ is located on point $x = \ell(r)$. 
When $r$ is processed by the Radius algorithm, then either $r$ is matched immediately with a pending request, and in that case $r$ has not been delayed, or it becomes a pending request. 
For the second case, let $r' \in \sigma$ be the first request arriving after $r$ at a point $y$ at distance at most $\rho_x$ from $x$. 
If such a request does not exist, then in the worst case, $r$ is matched at the very end and thus is delayed for a time $t(r_m) - t(r) = Y_r$. 
Otherwise, either $r$ was already matched when $r'$ arrived --- and then $r$ was delayed for a duration at most $t(r') - t(r) = Y_r$ --- or, we claim that $r$ and $r'$ are matched together by the Radius algorithm. 
Indeed, assume for a contradiction that $r$ and $r'$ are not matched together. 
Since $d(x,y) \le \rho_x$, it is necessarily because $r$ is matched with another pending request $r''$ located on a point $z$, such that $d(y,z) \le \rho_z$. 
But then, by the triangle inequality, we have $d(x,z) \le d(x,y) + d(y,z) \le \rho_x + \rho_z$ (see Figure \ref{figure:radius_intersection} for an example). 
This is impossible: the Radius algorithm would have matched $r$ and $r''$ together before the arrival of $r'$. 
Thus, in any of these cases, $r$ is delayed for a time at most $Y_r$.
\end{proof}

\begin{figure}[ht]
\input{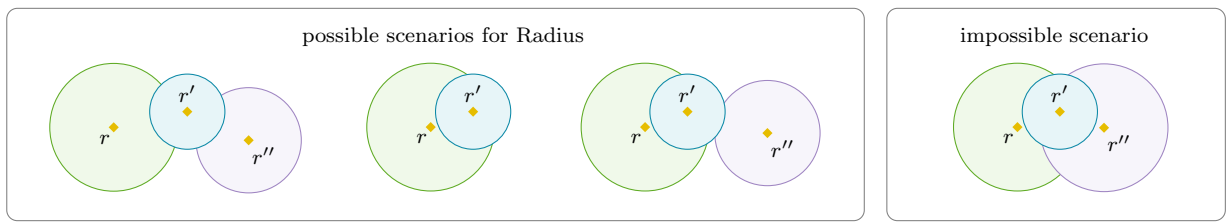}
\vspace{-16pt}
\caption{Example scenarios for the arrival of a new request $r'$ during the Radius execution. Firstly, $\location{r'}$ can belong to some pending request's ball. Secondly, the ball defined for $r'$ can intersect other requests' balls. However, by the definition of Radius, $\location{r'}$ cannot belong to the ball of more than one pending request.}
\label{figure:radius_intersection}
\end{figure}

\noindent Finally, we prove Lemma \ref{lemma:radius_upperbound}. 
\begin{proof}[Proof of Lemma \ref{lemma:radius_upperbound}]
Let $\sigma$ be a sequence of $m$ requests, and let $M $ denote the perfect matching output by the Radius algorithm. 
We split it into two sets $M_\text{nice}$ and $M_\text{late}$ of nice and late edges, respectively. 
The total cost $\rad(\sigma)$ of the solution is equal to $\cc(M_\text{nice}) + \cc(M_\text{late}) + \dc$, the sum of the connection cost $\cc(M_\text{nice})$ induced by the nice edges, the connection cost $\cc(M_\text{late})$ induced by the late edges and the total delay cost $\dc$. 

By Claim \ref{claim:few_late_edges}, there are at most $|\setofpoints| / 2$ late edges. 
\begin{equation} \label{eq:connection_cost_late}
    \cc(M_\text{late}) \le \frac{1}{2} \cdot |\setofpoints| \cdot d_{\max}.
\end{equation}
This bound is valid for any sequence, so, in particular, it is also valid in expectation. 

By Claim \ref{claim:connection_cost_nice}, we have $\cc(M_\text{nice}) \le \sum_{r \in \sigma} \rho_{\ell(r)}$. 
Thus, taking the average over all random sequences $\sigma = (r_1, \dots, r_m)$ consisting of $m$ requests, we obtain
\begin{align} \label{eq:connection_cost_nice}
    \E_\sigma^m\left[\cc(M_\text{nice})\right]
    &\le \E_\sigma^m\left[\sum_{r\in \sigma} \rho_{\ell(r)} \right]
    = \sum_{i = 1}^m \sum_{\anypoint \in \setofpoints} \radius{\anypoint} \cdot \prob(\ell(r_i) = x)
    = \sum_{i = 1}^m \sum_{\anypoint \in \setofpoints} \frac{\lambda_x}{\lambda(\setofpoints)} \cdot \rho_x \nonumber\\
    &\le m \sum_{x \in \setofpoints} \frac{\lambda_x}{\lambda(\setofpoints)} \cdot \rho_x.
\end{align}
We now show that total delay cost is in expectation at most 
\begin{equation} \label{eq:delay_cost}
    \E_\sigma^m[\dc]\le m \sum_{\anypoint \in \setofpoints} \frac{\lambda_x}{\lambda(\setofpoints)} \cdot \rho_x.
\end{equation}
For each point $x \in \setofpoints$, let $W_x \sim \expdistr{\lambda(\overline{B}(x,\rho_x))}$ be an exponential variable with parameter $\lambda(\overline{B}(x,\rho_x)) = \sum_{y \in \overline{B}(x,\rho_x)} \lambda_y$. 
Let us consider a random sequence $\sigma = (r_1, \dots, r_m)$ ordered by increasing arrival times, and let us assume that the $i$-th request $r_i$ is located at $x$. 
Then, by Proposition \ref{proposition:subset_exponential_waiting_time}, the (random) variable $Y_r$, as defined before Claim \ref{claim:delay_cost_radius}, follows the same distribution as $\min\{W_x, t(r_m) - t(r_i)\}$. 
Thus, using Claim \ref{claim:delay_cost_radius}, we prove inequality \eqref{eq:delay_cost} as follows:
\begin{align*}
    \E_\sigma^m\left[\dc\right]
    &\le \E_\sigma^m\left[\sum_{i=1}^m Y_{r_i} \right] &&\text{(Claim \ref{claim:delay_cost_radius})}\\
    &= \sum_{i=1}^m \sum_{x \in \setofpoints}  \prob(\ell(r_i)=x) \cdot \E_\sigma^m\left[Y_{r_i}\mid \ell(r_i) = x\right] &&\text{(linearity of expectation)}\\
    &= \sum_{i=1}^m \sum_{x \in \setofpoints}  \frac{\lambda_x}{\lambda(\setofpoints)} \cdot \E\left[\min\{W_x, t(r_m) - t(r_i)\}\right]\\
    &\le \sum_{i=1}^m \sum_{x \in \setofpoints}  \frac{\lambda_x}{\lambda(\setofpoints)} \cdot \E\left[W_x\right]\\
    &\le \sum_{i=1}^m \sum_{x \in \setofpoints}  \frac{\lambda_x}{\lambda(\setofpoints)} \cdot \rho_x &&\text{(Observation \ref{observation:lowerboundingscheme})}\\
    &= m \sum_{x\in\setofpoints}  \frac{\lambda_x}{\lambda(\setofpoints)} \cdot \rho_x.
\end{align*}
Finally, we obtain the bound claimed in Lemma \ref{lemma:radius_upperbound}, by summing equations \eqref{eq:connection_cost_late}, \eqref{eq:connection_cost_nice} and \eqref{eq:delay_cost}. This concludes the proof. 
\end{proof}

\section{Extension to general delay costs}
\label{sec:general_delay}

\noindent In this section, we study a generalization of the MPMD problem where the delay cost function is not required to be linear. 
In this version of the problem, referred to as $f$-MPMD, the decision of matching a request can be postponed for
time $t$ at a delay cost of $f(t)$, where $f$ is the \emph{delay cost function}. 
We require this function to be \emph{positive} (otherwise some solutions may have negative value), and \emph{non-decreasing}. 
In Appendix \ref{app:general_delay}, we show that w.l.o.g., we can assume $f(0) = 0$, i.e., if a request is directly matched at its arrival time, no delay cost is incurred. 

This more general version of MPMD have been investigated for the classic online adversarial model among others by Azar et al.\ \cite{azar2021min} and Liu et al.\ \cite{liu2018impatient}. 
Their works suggests that in general, the $f$-MPDM problem is even more challenging than the original MPMD problem. 
For instance, in \cite{azar2021min}, Azar et al.\ considered a special type of concave delay cost function, and showed that even for the single-point metric, obtaining a constant competitive algorithm is a non-trivial task (whereas for the linear case, the optimal algorithm simply matches two consecutive requests). 
Liu et al.\ \cite{liu2018impatient} showed that under some natural requirements for function $f$, any deterministic online algorithm for $f$-MPMD on a $k$-points metric must have a competitive ratio $\Omega(k)$.  

In the online stochastic Poisson arrival model, we show in Theorem \ref{theorem:general_delay} that our Greedy and Radius algorithms for MPMD can be adapted to $f$-MPMD, and that their corresponding ratios-of-expectations remain a constant, which depends on $f$. 

The adaptation of the Greedy algorithm for $f$-MPMD is quite straightforward: when the sum of the delay cost of two pending requests exceeds their distance, match them together, i.e., match pending requests $r$ and $r'$ at time $t$ whenever $f(t - t(r)) + f(t - t(r')) \ge d(\ell(r), \ell(r'))$. 

The Radius algorithm works in the general case exactly as in the linear case, but using the following generalized definition of radius.  

\begin{definition} \label{definition:general_radius}
    Given the positive and non-decreasing delay cost function $f$, for any point $x \in \setofpoints$, define its radius $\rho_x$ as the smallest value $u \in \mathbb{R}^+ \cup \{\infty\}$ such that
    \begin{equation*}
        u \ge \E[f(\overline{W}(x,u))],
    \end{equation*}
    where $\overline{W}(x,u)$ is an exponential variable of parameter $\lambda(\overline{B}(x,u)) := \sum_{y \in \setofpoints: d(x,y) \le u} \lambda_y$.
\end{definition}
\noindent Since functions $u \mapsto \lambda(\overline{B}(x,u))$ and $f$ are both non-decreasing, the function $u \mapsto \E[f(\overline{W}(x,u))]$ is non-increasing. 
This implies that the radius of each point is well-defined and unique. 
Moreover, since $\E[\overline{W}(x,u))] = 1 / \lambda(\overline{B}(x,u))$, in the case when $f(t) = t$, this definition coincides with our initial Definition \ref{definition:radius}. 
Similarly as presented in Observation \ref{observation:lowerboundingscheme}, it is easy to see that $\E[f({W}^\circ(x,\rho_x))] \ge \rho_x$, where ${W}^\circ(x,\rho_x)$ is a random variable of parameter $\lambda({B}^\circ(x,\rho_x)) := \sum_{y \in  \setofpoints: d(x,y) < \rho_x} \lambda_y$. 

We now give the main result of this section. 
\begin{theorem} \label{theorem:general_delay}
    Consider an instance of the $f$-MPMD problem such that $\E[f(X)] < \infty$, where $X \sim \expdistr{\lambda(\setofpoints)}$ is an exponential variable of parameter $\lambda(\setofpoints) := \sum_{x \in \setofpoints} \lambda_x$.
    Then, both the Greedy and Radius algorithms achieve ratio of expectations of $O(K_f)$, where
    \begin{equation*}
        K_f := \max_{\mu > 0} \left\{\dfrac{\E[f(X)]}{\E[\min(f(X'), \E[f(X)])]},\text{ where } X \sim \expdistr{\mu} \text{ and } X' \sim \expdistr{2\mu}\right\}.
    \end{equation*}
\end{theorem}
The condition $\E[f(X)] < \infty$ means that the expected delay cost corresponding to the duration between any two consecutive requests in the random sequence is finite. 
Without this assumption, even the expected cost of the optimal offline matching, over all random sequence of length $m=2$, is infinite. 
Notice that under this assumption, we have that the radius (as defined in Definition \ref{definition:general_radius}) for each point, is finite. 

The full proof of Theorem \ref{theorem:general_delay} is presented in Appendix \ref{app:general_delay} and follows the same framework as for the linear case. 
The upper bounds on the expected costs of Greedy and Radius are the same as the upper bounds in Lemmas \ref{lemma:greedy_upperbound} and \ref{lemma:radius_upperbound}. 
The factor $O(K_f)$ in the ratio of expectations comes from the analysis of the expected cost of the optimal offline solution. 
Essentially, we show that in a random sequence $\sigma$, for any request located at point $x$, its (generalized) minimum total cost in $\sigma$ is in expectation $\rho_x / K_f$.\footnote{The original definition for linear delays is given before Claim \ref{claim:opt_tilde}. Here, we need to adapt this notion to the new delay cost function $f$.}

There are several natural functions $f$ for which we can give an explicit value of $K_f$. 
For example, in the linear case, that is when $f(t) = t$, it follows directly from Claim \ref{claim:exponential-2} that $K_f = 2 / (1 - e^{-2})$.
More generally, in the case when $f(t) = t^\alpha$, for some positive constant $\alpha \ge 0$, which has been studied in \cite{liu2018impatient}, we show in Appendix \ref{app:general_delay} that 
\begin{proposition} \label{proposition:delay_alpha}
    $K_f = e^{O({\alpha})}$ when $f(t) = t^\alpha$, with $\alpha \ge 0$. 
\end{proposition}
Notice that for any polynomial function $f$ (so in particular for all functions $t \mapsto t^\alpha$), and any exponential variable $X$, we have $\E[f(X)] < \infty$. 
Thus, we obtain the following corollary. 
\begin{corollary}
    Both the Greedy and the Radius algorithms achieve a ratios-of-expectations of $e^{O({\alpha})}$ for the $t^\alpha$-MPMD problem. 
\end{corollary}

\section{Paying penalties to clear pending requests}
\label{section:mpmdfp}
In this section, we consider a variant of MPMD called MPMDfp \cite{emek2016online}, where it is allowed to clear any request by paying a fixed penalty $p > 0$.
For this problem, we propose the following algorithm $\algor$, that works similarly to Radius, obtaining a constant ratio of expectations. 

Given the metric space ($\setofpoints, d$), define $\setofpoints^{(1)} = \{x \in \setofpoints: \radius{x} < p\}$ and $\setofpoints^{(2)} = \{x \in \setofpoints: \radius{x} \ge p\}$ (where $\radius{x}$ is the radius of point $x \in \setofpoints$ as defined in Definition \ref{definition:radius}). Suppose that at time $t$, a new request $r$ arrives. Then, our algorithm performs the following actions depending on whether $\ell(r)\in \setofpoints^{(2)}$ or $\ell(r)\in \setofpoints^{(1)}$:
\begin{itemize}
    \item[-] Suppose $\ell(r) \in \setofpoints^{(2)}$. If there exists a pending request $r'$ located at point $y \in \setofpoints^{(1)}$ and $x \in \overline{B}(y, \radius{y})$, then match $r$ with $r'$. Otherwise, clear $r$. 
    \item[-] Suppose $\ell(r) \in \setofpoints^{(1)}$. Apply the Radius algorithm to match this request. 
\end{itemize}
Notice that there possibly exists an odd number of late requests (due to clearing an odd number of requests arriving at points $\setofpoints^{(2)}$). In that case, $\algor$ has to clear the last request $r_m$ even when $\ell(r_m)\in\setofpoints^{(1)}$. 

\begin{theorem}
    For MPMDfp in the Poisson arrival model, $\algor$ achieves a ratio of expectations of $8 / (1 - e^{-2})$. 
\end{theorem}
The full proof of this theorem is presented in Appendix \ref{app:mpmdfp}, and follows the same framework as sketched in Section \ref{section:algorithms}. 
On the one hand, we lower bound $\E_{\sigma}^m[\optim(\sigma)]$ by
\begin{equation*}
    m \cdot \frac{1 - e^{-2}}{4} \sum_{\anypoint \in \setofpoints} \frac{\lambda_x}{\lambda(\setofpoints)} \cdot \min\{\radius{\anypoint}, p\};
\end{equation*}
on the other hand, we upper bound the expected cost produced by $\algor$ by
\begin{equation*}
    \left(2m \sum_{\anypoint \in \setofpoints} \frac{\lambda_x}{\lambda(\setofpoints)} \cdot \min\{\radius{\anypoint}, p\}\right) + \frac{1}{2} \cdot |\setofpoints| \cdot d_{max}.
\end{equation*}

\paragraph{Remark.} In \cite{emek2016online}, Emek et al.\ showed a reduction from MPMD to MPMDfp in the case where $\smash{p < 2 \cdot d_{\max}}$. They argue that any online algorithm for the instance of MPMD obtained by this reduction can be turned into an algorithm for the original instance of MPMDfp, while only loosing a factor 2 in the competitive ratio. 

\section{Extension to asymmetric distance costs}
\label{section:asymmetric}
In this section, we consider a generalized version of this problem where the given metric space $\metricspace = (\setofpoints, d)$ is asymmetric, i.e., it is not necessary $d(x, y) = d(y, x)$ for any two different points $x, y \in \setofpoints$. 
However, the distance function still satisfies the triangle inequality, i.e., for any three points $x, y, z \in \setofpoints$, we have $d(x, y) + d(y, z) \ge d(x, z)$. 
Under such assumption, the distance cost of a pair $(r, r')$ is defined as
\begin{equation*}
    \frac{d(\ell(r),\ell(r')) + d(\ell(r'), \ell(r))}{2}.
\end{equation*}

For this asymmetric distance version of MPMD in stochastic model, the Greedy algorithm, matching any two pending requests $r, r'$ into a pair when their total delay cost is at least $$\frac{d(\ell(r), \ell(r')) + d(\ell(r'), \ell(r))}{2},$$ still achieves a constant ratio of expectations. 
\begin{theorem}
    For asymmetric MPMD in the Poisson arrival model, $\greedy$ achieves a ratio of expectations of $16 / (1 - e^{-2})$. 
\end{theorem}
See Appendix \ref{app:asymmetric} for a full proof of this theorem.

\section{Conclusion}
\label{section:conclusion}
In this paper, we considered the online problem of Min-cost Perfect Matching with Delays (MPMD) with additional stochastic assumption on the sequence of the input requests. 
In the case where the requests follow a Poisson arrival process, we presented two simple deterministic online algorithms with constant ratio of expectations.
In particular, we observed that the cost of the optimal offline solution is proportional to the number of requests in the sequence, and gave a tight (up to a constant factor independent from the instance) estimation of the constant of proportionality.

In the following text, we briefly discuss some potential future directions. 

\subsection{The bipartite case in the Poisson arrival model}
Previously, the bipartite version of MPMD (i.e., MBPMD) has been considered in the adversarial model \cite{ashlagi2017min} where each request has a color, either red or blue, and only requests of different colors can be matched into a pair\footnote{For an application, imagine that the red requests come from customers and the blue ones represent the suppliers.}. 
In an equivalent definition, given the metric space $\metricspace = (\setofpoints, \distancesymbol)$, the points of $\setofpoints$ are partitioned into two subsets $A$ and $B$, such that the requests arriving at points $A$ can only be matched with requests from points $B$. 
Ashlagi et al.\ \cite{ashlagi2017min} proposed two O($\log n$)-competitive randomized online algorithms for this problem.
Besides, they established a lower bound of $\Omega(\sqrt{\log n / \log \log n})$ on the competitive ratio of any online algorithm. 
Note that the MBPMD problem can be seen as a special case of the \emph{non-metric} perfect matching problem with delays, where the connection cost function $d: \setofpoints \times \setofpoints \rightarrow \R_+ \cup \{\infty\}$ can have infinite values and is no longer assumed to satisfy the triangle inequality. 

A natural direction would be to explore MBPMD in the Poisson arrival model. 
Unfortunately, the following observation establishes an initial difficulty: the expected cost of the offline optimal algorithm, on random sequence of length $m$, cannot be upper bounded by $O(m)$. 

\begin{lemma}
    Let $\setofpoints = \{a, b\}$, and assume that the connection cost $d$ satisfies $d(a, a) = d(b, b) = \infty$ and $d(a, b) = d(b, a) = 0$. 
    In the Poisson arrival model, the red (resp.\ blue) requests arrive at point $a$ (resp.\ $b$) with a Poisson arrival rate $\lambda_a = 1 / 2$ (resp.\ $\lambda_b = 1 / 2$).
    Let $\sigma$ be a random request sequence of length $m$. 
    Consider the algorithm $\algor$ that matches any two pending requests from $a$ and $b$ greedily. 
    Then, $\E_\sigma^m[\algor(\sigma)] = \Omega(m\sqrt{m})$.
\end{lemma}

\begin{proof}
First, notice that $\algor$ is the optimal algorithm in both online and offline settings. 
Given any $i \in \{1, \dots, m\}$, let $P_i$ denote the (random) number of requests still pending right after the arrival of $i$-th request. 
Let $W_i$ denote the time duration between the arrival times of $i$-th and ($i+1$)-th requests. 
Since $\lambda_a = \lambda_b = 1 / 2$, we have $W_i \sim \expdistr{1}$. 
The total cost of $\algor$, which corresponds to its total delay cost, is
\begin{equation*}
    \algor(\sigma)=\sum_{i=1}^{m-1} P_i \cdot W_i.
\end{equation*}
Since $W_i$ is independent from $P_i$ according to Definition \ref{definition:centralized_poisson_arrival}, for any $i \in \{1, \dots, m\}$, we have
\begin{equation*}
    \E_\sigma^m[\algor(\sigma)] = \sum_{i=1}^{m-1} \E_\sigma^m[P_i] \cdot \E_\sigma^m[W_i] = \sum_{i=1}^{m-1} \E_\sigma^m[P_i].
\end{equation*}
Note that $\{P_i\}_i$ is the translation distance of a (uniform) one-dimensional random walk. 
It is known that: 
\begin{equation*}
    \E[P_i] \sim \sqrt{\frac{2}{\pi}}\cdot\sqrt{i}. 
\end{equation*}
This implies that 
\begin{equation*}
    \E_\sigma^m[\algor(\sigma)] = \sum_{i=1}^{m-1} \E_\sigma^m[P_i] = \Omega\left(\sum_{i=1}^{m-1}\sqrt{i}\right) = \Omega(m\sqrt{m}).
\end{equation*}
\end{proof}
In the bipartite case, since two requests on the same point cannot be matched together, pending requests on the same point will accumulate over time, forming queues. 
In particular, the waiting time of a request depends on the size of the queue at the time of its arrival. By the previous lemma, we know that the size of the queues will grow to infinite, and in particular, as time passes by, the delay cost will become arbitrarily larger than the connection cost. 
This suggests that a simple algorithm that matches any two pending requests, no matter how large their connection cost is, is essentially the best possible online algorithm for this stochastic version of MBPMD. 

\subsection{$k$-way min-cost perfect matching with delays}
Another direction, that was introduced by \cite{melnyk2021online} for the online adversarial model, would be to consider a generalized $k$-way min-cost perfect matching with delays ($k$-MPMD) in the stochastic input model, where each pair (a.k.a., $k$-tuple) consists of $k$ different requests ($k \ge 2$ is an arbitrary integer).
Note that such $k$-MPMD problem indeed has real applications from ride-sharing taxi platforms (when a taxi picks up $k$ passengers from different locations for one ride) and online gaming platforms (when a gaming session consists of $k$ different players). 
To attack this version of the MPMD problem, one should first come out with a suitable notion of ``connection cost'' of a $k$-set. 
This might be for instance measured by the maximum distance between any two requests of that set, the average distance, the weight of a minimum spanning tree, etc.  

\subsection{Online network design problems with delays in the Poisson arrival model}
Besides online matching with delays, another online network design problem called multi-level aggregation is also considered in this Poisson arrival model \cite{mari2024online}.
There also exists online algorithm with constant ratio of expectations. 
For the other online network design problems, such as service with delays (and its generalization called $k$-services with delays), facility location with delays, Steiner tree/forest with delays etc, does there also exist online algorithm with constant ratio of expectations in the Poisson arrival model?

\section*{Acknowledgement}
This work was partially supported by the ERC CoG grant TUgbOAT no 772346, Polish NCN grant no 2020/37/B/ST6/04179, no 2022/45/B/ST6/00559 and no 2024/53/N/ST6/04119. 
For the purpose of Open Access, the author has applied CC-BY public copyright license to any Author Accepted Manuscript (AAM) version arising from this submission.

\bibliographystyle{siam}
\bibliography{bibliography}

\appendix 

\section{Elementary proof of Proposition \ref{proposition:subset_exponential_waiting_time}}
\label{app:subset_distribution_elementary}
In this section, we present a proof of Proposition \ref{proposition:subset_exponential_waiting_time} that is more elementary than the one included in the main part of this paper. 
First, let us restate this proposition.

\subsetdistribution*

It is known that we can look at a Poisson process as the limit of a Bernoulli process. 
In our case, it means that we can proceed in two steps. 
First, we have to divide the timeline $[0,\infty)$ into tiny sub-intervals of length $\delta$. 
Then, for each of them, toss a coin with success probability $\lambda(\setofpoints) \cdot \delta$ to decide whether it contains a request arrival. 
Finally, in our process, we also have to draw the locations. 
By Definition \ref{definition:centralized_poisson_arrival}, we do it independently for each request with the probability of it appearing at point $\anypoint$ equal to $\waitingparam[\anypoint] / \waitingparam(\setofpoints)$ for each $x \in \setofpoints$.

Now, w.l.o.g., let us assume that $i = 1$ and $\arrival{r_1} = 0$, as we only look at the future requests here. 
Let $W$ denote the waiting time for the first request arriving after $r_1$ at any point of $\S$. 
Our aim is to find the tail distribution $\bar{F}_W(t)$ of $W$. 
For this purpose, we notice that at any given time $t > 0$, it holds that a request can arrive between $t$ and $t + \delta$ with probability $\lambda(\setofpoints) \cdot \delta$. 
Moreover, the chances of its location belonging to $\S$ are equal to $\lambda(\S) / \lambda(\setofpoints)$. Thus, we can write that
\begin{equation*}
    \bar{F}_W(t + \delta) = \bar{F}_W(t) - \bar{F}_W(t) \cdot \lambda(\setofpoints) \cdot \delta \cdot \dfrac{\lambda(\S)}{\lambda(\setofpoints)}
\end{equation*}
which is equivalent to
\begin{equation*}
    \bar{F}_W(t+\delta) - \bar{F}_W(t) = -\bar{F}_W(t) \cdot \delta \cdot \lambda(\S).
\end{equation*}
After dividing both sides by $\delta$ and taking the limit $\delta \to 0$, we obtain
\begin{equation*}
    \bar{F}_W'(t) = -\lambda(\S) \cdot \bar{F}_W(t).
\end{equation*}
Together with condition $\bar{F}_W(0) = 1$, it implies that $\bar{F}_W(t) = e^{-\lambda(\S)t}$. 
Thus, the cumulative distribution function of $W$ equals $1 - e^{-\lambda(\S)t}$, which means that $W$ has an exponential distribution with parameter $\lambda(\S)$.

\section{Missing proofs in Section \ref{sec:general_delay}}
\label{app:general_delay}
We show that w.l.o.g., we can assume that the delay cost function $f$ satisfies $f(0) = 0$.

\begin{claim}
    Let $f$ be any positive and non-decreasing function, and define a new function $\hat{f}$ with $\hat{f}(t) = f(t) - f(0)$.
    Suppose that there exists a $C$-competitive algorithm $\widehat{\algor}$ for $\hat{f}$-MPMD, for some constant $C > 0$. 
    Then, $\widehat{\algor}$ is also $C$-competitive for $f$-MPMD.  
\end{claim}

\begin{proof}
For any input sequence $\sigma$ of length $m$, we have $\algor(\sigma) = \widehat{\algor}(\sigma) + m \cdot f(0)$, where $\algor(\sigma)$ and $\widehat{\algor}(\sigma)$ respectively denote the cost of the algorithm with delay cost functions $f$ and $\hat{f}$. 
Similarly, it is easy to see that $\optim(\sigma) = \widehat{\optim}(\sigma) + m \cdot f(0)$, where 
$\optim(\sigma)$ and $\widehat{\optim}(\sigma)$ denote the total cost of the optimal offline solutions, with respective delay cost functions $f$ and $\hat{f}$. 
In particular, 
\begin{equation*}
    \frac{\algor(\sigma)}{\optim(\sigma)} = \frac{\widehat{\algor}(\sigma) + m \cdot f(0)}{\widehat{\optim}(\sigma) + m \cdot f(0)} \le \frac{\widehat{\algor}(\sigma)}{\widehat{\optim}(\sigma)}
\end{equation*}
which proves the claim. 
\end{proof} 

\subsection{Proof of Theorem \ref{theorem:general_delay}}
The proof is similar as the proof of Theorems \ref{main:greedy} and \ref{main:radius}. We first prove the following upper bounds on the expected cost of the solutions of Radius and Greedy:
\begin{equation} \label{eq:radius_general_delay}
    \E_{\sigma}^m[\rad(\sigma)] \le \left(2m \sum_{\anypoint \in \setofpoints} \frac{\lambda_x}{\lambda(\setofpoints)} \cdot \radius{\anypoint}\right) + \frac{1}{2} \cdot |\setofpoints| \cdot d_{\max},
\end{equation}
\begin{equation} \label{eq:greedy_general_delay}
    \E_{\sigma}^m[\greedy(\sigma)] \le \left(4m \sum_{\anypoint \in \setofpoints} \frac{\lambda_x}{\lambda(\setofpoints)} \cdot \radius{\anypoint}\right) + 2|\setofpoints| \cdot \left(d_{\max}+\frac{1}{\lambda(\setofpoints)}\right),
\end{equation}
where the radii $\{\rho_x\}_{x \in \setofpoints}$ are defined in Definition \ref{definition:general_radius}. 
On the other hand, we establish the following lower bound on the expected cost of the optimal offline solution: 
\begin{equation} \label{eq:opt_general_delay}
    \E_\sigma^m[\optim(\sigma)] \ge \frac{m}{2K_f} \sum_{\anypoint \in \setofpoints} \frac{\lambda_x}{\lambda(\setofpoints)} \cdot \rho_x.
\end{equation}
The proofs of inequalities \eqref{eq:radius_general_delay} and \eqref{eq:greedy_general_delay} are almost identical as the proofs of Lemmas \ref{lemma:radius_upperbound} and \ref{lemma:greedy_upperbound}. 
Here we only show how to adapt the framework developed in Section \ref{section:proof_radius} to obtain inequality \eqref{eq:radius_general_delay}. 

We use the same definition of nice and late requests. In particular, Claim \ref{claim:few_late_edges} still holds, i.e., the connection cost induced by late edges is at most $|\setofpoints| \cdot d_{\max} / 2$. 
Since the algorithm is the same, Claim \ref{claim:connection_cost_nice} still holds: for any sequence $\sigma$, the connection cost induced by all nice edges is at most $\sum_{\request \in \sigma} \radius{\location{\request}}$. 
To bound the delay cost of nice and late requests, we define $Y_r$, for each request $\request \in \sigma$ exactly as in Section \ref{section:proof_radius}:
\begin{equation*}
    Y_r := \min\left\{t(r_m) - t(r), \min_{r' \in \sigma} \left\{t(r') - t(r) \mid t(r') > t(r) \text{ and } d(\ell(r'), \ell(r)) \le \radius{\location{\request}} \right\} \right\}.
\end{equation*}
Again, Claim \ref{claim:delay_cost_radius} holds: each request is delayed by the Radius algorithm for a duration at most $Y_r$. 
The only slight difference with the proof from Section \ref{section:proof_radius} comes when we establish an upper bound on the expected total delay cost of the solution. 

Similarly, for each point $x \in \setofpoints$, let $\overline{W}(x,\radius{x}) \sim \expdistr{\lambda(\overline{B}(x,\rho_x))}$ be an exponential variable of parameter $\lambda(\overline{B}(x,\rho_x)) = \sum_{y \in \overline{B}(x,\rho_x)} \lambda_y$. 
Let us consider a random sequence $\sigma = (r_1, \dots, r_m)$ ordered with increasing arrival times, and let us assume that the $i$-th request $r_i$ is located at point $x$. 
By Proposition \ref{proposition:subset_exponential_waiting_time}, the (random) variable $Y_r$, follows the same distribution as $\min\{\overline{W}(x, \rho_x), t(r_m) - t(r_i)\}$. 
Using the general definition of radius (Definition \ref{definition:general_radius}), instead of Observation \ref{observation:lowerboundingscheme}, we obtain:
\begin{align*}
    \E_\sigma^m\left[\dc\right]
    &\le \E_\sigma^m\left[\sum_{i=1}^m f(Y_{r_i})\right] \\
    &= \sum_{i = 1}^m \sum_{x \in \setofpoints} \prob(\ell(r_i) = x) \cdot \E_\sigma^m\left[f(Y_{r_i})\mid \ell(r_i) = x\right] \\
    &= \sum_{i = 1}^m \sum_{x \in \setofpoints}  \frac{\lambda_x}{\lambda(\setofpoints)} \cdot \E\left[f(\min\{\overline{W}(x,\rho_x), t(r_m) - t(r_i))\}\right]\\
    &\le \sum_{i = 1}^m \sum_{x \in \setofpoints}  \frac{\lambda_x}{\lambda(\setofpoints)} \cdot \E\left[f(\overline{W}(x,\rho_x)) \right] &&\text{($f$ is non-decreasing)}\\
    &\le \sum_{i = 1}^m \sum_{x \in \setofpoints}  \frac{\lambda_x}{\lambda(\setofpoints)} \cdot \rho_x &&\text{(Definition \ref{definition:general_radius})}\\
    &= m \sum_{x \in \setofpoints}  \frac{\lambda_x}{\lambda(\setofpoints)} \cdot \rho_x,
\end{align*}
which completes the proof of inequality \eqref{eq:radius_general_delay}.

\paragraph{Proof of inequality \eqref{eq:opt_general_delay}.} 
\noindent The proof is similar as the proof of Lemma \ref{lemma:lowerboundingscheme}. 
The main difference is that instead of using Claim \ref{claim:exponential-2}, we use the definition of $K_f$. 

Similarly, given a request sequence $\sigma$ and any request $r \in \sigma$, we define the \emph{minimum total cost of $r$ in $\sigma$} as 
\begin{equation*}
    c(\sigma, r) := \min_{r' \in \sigma, r' \neq r} \left\{d(\ell(r),\ell(r')) + f(t(r) - t(r'))\right\}.
\end{equation*}
With the same proof as Claim \ref{claim:opt_tilde}, we have $\optim(\sigma) \ge \frac{1}{2}\sum_{r \in \sigma} c(\sigma, r)$. 
\begin{claim}
\label{claim:expected_cost_x_general}
Given any sequence $\sigma$, we order the requests in $\sigma = (r_1, \dots, r_m)$ according to their arrival times. 
Then, for any point $x \in \setofpoints$ and any index $i \in \{1, \dots, m\}$, the expected minimum cost of the $i$-th request $r_i$ in a random sequence $\sigma$, assuming that $r_i$ is located at point $x$, is
\begin{equation*}
    \E_\sigma^m[c(\sigma, r_i) \mid \ell(r_i) = \anypoint] \ge  \frac{\rho_x}{K_f}.
\end{equation*}
\end{claim}

\begin{proof}
The proof is similar as the proof of Claim \ref{claim:expected_cost_x}. 
In particular, using the same argumentation, it is enough to establish the bound of Claim \ref{claim:expected_cost_x_general} for extended random request sequences $\overline{\sigma}$. 
Exactly as before, we can prove that 
\begin{equation*}
    \E_\sigma^m[c(\sigma, r_i) \mid \ell(r_i) = x] \ge \E\Big[\min\big\{f(\min\{W^+, W^-\}),\rho_x\big\}\Big],
\end{equation*}
where $W^+$ and $W^-$ are two independent random exponential variables of parameter $\lambda(B^\circ(x,\rho_x))$. 
In particular, the random variable $W = \min\{W^+, W^-\}$ is an exponential variable of parameter $2\lambda(B^\circ(x, \rho_x))$. 
As we observed after Definition \ref{definition:general_radius}, we know that 
$\E[f(W^+)] \ge \rho_x$. 
Let us define $X$ and $X'$ two independent exponential variables of parameters $1 / \rho_x$ and $2 / \rho_x$ respectively. 
We have $\rho_x = \E[f(X)]$.
Using the definition of $K_f$, we have:
\begin{align*}
    \E_\sigma^m[c(\sigma, r_i) \mid \ell(r_i) = x]
    &\ge \E[\min\{f(W),\rho_x\}] \ge \E\big[\min\{f(X'),\rho_x\}\big]\\
    &= \frac{\E\Big[\min\big\{f(X'),\E[f(X)]\big\}\Big]}{\E[f(X)]} \cdot \E[f(X)] \ge \frac{\rho_x}{K_f},
\end{align*}
where the second inequality holds since ${2}/{\rho_x} \ge 2\lambda(B^\circ(x,\rho_x)))$ and $t \mapsto \min\{f(t), \rho_x\}$ is non-decreasing. 
This concludes the proof of Claim \ref{claim:expected_cost_x_general}. 
We obtain inequality \eqref{eq:opt_general_delay} from this claim, similarly as in the proof of Lemma \ref{lemma:lowerboundingscheme}. 
This concludes the proof of Theorem \ref{theorem:general_delay}.
\end{proof}

\subsection{Proof of Proposition \ref{proposition:delay_alpha}}
In this section, we prove that $K_f = e^{O({\alpha})}$ when $f(t) = t^\alpha$ and $\alpha > 0$. 
Let $\mu > 0$ and define $X$ and $X'$ two independent exponential variables of parameters $\mu$ and $2\mu$ respectively. 
We have
\begin{equation*}
    \E[f(X)]=\int_0^\infty t^\alpha \mu e^{-\mu t}dt=\frac{\Gamma(\alpha+1)}{\mu^\alpha}
\end{equation*}
where $\Gamma(\cdot)$ is the Gamma function. We obtain the follows expression.  
\begin{align*}
    \E[\min\{f(X'),\E[f(X)]\}] = \int_0^{\sqrt[\alpha]{\E[f(X)]}} t^\alpha \cdot 2\mu e^{-2\mu t}dt + {\E[f(X)]} \cdot \prob\left(f(X') > \E[f(X)]\right) 
\end{align*}
One can express the first term using the incomplete Gamma function. 
However for simplicity, we only lower this term by $0$, and we obtain
\begin{align*}
    \frac{\E[\min\{f(X'),\E[f(X)]\}]}{{\E[f(X)]}}
    &\ge \prob\left(f(X') > \E[f(X)]\right) 
    = \prob\left(X'>f^{-1}\left(\E[f(X)]\right)\right) \\
    &= \prob\left(X' > \sqrt[\alpha]{\frac{\Gamma(\alpha+1)}{\mu^\alpha}}\right)
    = \text{exp}\left(-2\mu\sqrt[\alpha]{\frac{\Gamma(\alpha+1)}{\mu^\alpha}}\right).
\end{align*}
Thanks to Stirling formula, when $\alpha$ goes to infinity, we have
\begin{equation*}
    \sqrt[\alpha]{\Gamma(\alpha+1)} \sim \frac{\alpha}{e} \cdot (1 + o(1)).
\end{equation*}
This implies that 
\begin{equation*}
    K_f \ge \frac{\E[\min\{f(X'),\E[f(X)]\}]}{{\E[f(X)]}} \ge \prob\left(f(X')>\E[f(X)]\right) \ge e^{2\alpha(1+o(1))/e} = e^{O(\alpha)}, 
\end{equation*}
which concludes the proof of the proposition. 

\section{Missing proofs in Section \ref{section:mpmdfp}}
\label{app:mpmdfp}

\paragraph{Lower bounding $\E_{\sigma}^m[\optim(\sigma)]$. }
Given any input sequence $\sigma$ and any request $\request \in \sigma$, define 
\begin{equation*}
    c'(\sigma, r) = \min\{c(\sigma, r), p\},
\end{equation*}
where $c(\sigma, r) := \min_{r' \in \sigma, r' \neq r} \left\{d(\ell(r),\ell(r')) + |t(r) - t(r')|\right\}$ is the minimum total cost of $r$, defined in the proof of Lemma \ref{lemma:lowerboundingscheme}. 
Note that if $\request$ is cleared in $\optim(\sigma)$, then a total cost of $p \ge c'(\sigma, r)$ is incurred for dealing with such request; 
otherwise, $\request$ is matched with another request $r'$ in $\optim(\sigma)$, and hence a cost of $d(\location{r}, \location{r'}) + |\arrival{r} - \arrival{r'}| \ge c(\sigma, r) \ge c'(\sigma, r)$ is incurred for producing such pair.
Therefore, we have $\optim(\sigma) \ge \frac{1}{2}\sum_{\request \in \sigma} c'(\sigma, r)$ for any input sequence $\sigma$.
Similar to the proof of Lemma \ref{lemma:lowerboundingscheme}, now we only need to prove 
\begin{equation*}
\E_{\sigma}^m[c'(\sigma, r)] \ge \frac{1 - e^{-2}}{2} \sum_{\anypoint \in \setofpoints} \frac{\lambda_x}{\lambda(\setofpoints)} \cdot \min\{\radius{\anypoint}, p\}.
\end{equation*}
Thanks to the definitions of extended sequence $\overline{\sigma}$, $W^-$ and $W^+$ introduced in the proof of Claim \ref{claim:expected_cost_x}, we further have
\begin{equation*}
    c'(\overline{\sigma}, r_0) \ge \min\Big\{\min\{W^-, W^+\}, \radius{\anypoint}, p\Big\} = \min\Big\{\min\{W^-, W^+\}, \min\{\radius{\anypoint}, p\}\Big\}.
\end{equation*}
Note that given any metrical point $x \in \setofpoints^{(1)}$, i.e., $\radius{x} < p$ and hence $\min\{\radius{x}, p\} = \radius{x}$, by the definition of radius (Definition \ref{definition:radius}), 
\begin{equation*}
    \frac{1}{\lambda(B^\circ(x, \radius{x}))} \ge \radius{x}, \text{ i.e., } \lambda(B^\circ(x, \radius{x})) \cdot \radius{x} \le 1;
\end{equation*}
given any metrical point $x \in \setofpoints^{(2)}$, i.e., $\radius{x} \ge p$ and hence $\min\{\radius{x}, p\} = p$,
\begin{equation*}
    \frac{1}{\lambda(B^\circ(x, \radius{x}))} \ge \radius{x} \ge p, \text{ i.e., } \lambda(B^\circ(x, \radius{x})) \cdot p \le 1.
\end{equation*}
In summary, for any metrical point $x \in X$, we have
\begin{equation*}
    A := \lambda(B^\circ(x, \radius{x})) \cdot \min\{\radius{x}, p\} \le 1.
\end{equation*}
By Claim \ref{claim:exponential-2}, with $a = \min\{\radius{x}, p\}$, $\mu = 2\lambda(B^\circ(x, \radius{x}))$, we thus have
\begin{equation*}
    \E_{\overline{\sigma}}[c'(\overline{\sigma}, r_0) \mid \ell(r_0) = \anypoint] 
    \ge \E_{\overline{\sigma}}\Big[\min\big\{\min\{W^-, W^+\}, \min\{\radius{\anypoint}, p\}\big\}\Big] 
    = \frac{1 - e^{-2A}}{2A} \cdot \min\{\radius{\anypoint}, p\}.
\end{equation*}
Recall that $x \mapsto \frac{1 - e^{-x}}{x}$ is a strictly decreasing function of $x > 0$ and $A \le 1$.
We have $\frac{1 - e^{-2A}}{2A} \ge \frac{1 - e^{-2}}{2}$ and hence
\begin{equation*}
    \E_{\overline{\sigma}}[c'(\overline{\sigma}, r_0) \mid \ell(r_0) = \anypoint] 
    \ge \frac{1 - e^{-2}}{2} \cdot \min\{\radius{\anypoint}, p\}.
\end{equation*}
Thanks to Proposition \ref{proposition:minimum}, we immediately have
\begin{equation*}
    \E_{\overline{\sigma}}[c'(\overline{\sigma}, r_0)] \ge \frac{1 - e^{-2}}{2} \sum_{\anypoint \in \setofpoints} \frac{\lambda_x}{\lambda(\setofpoints)} \cdot \min\{\radius{\anypoint}, p\}.
\end{equation*}

\paragraph{Upper bounding $\E_{\sigma}^m[\algor(\sigma)]$. } Given any input sequence $\sigma$, we classify the requests into two groups according to their locations
\begin{equation*}
    \sigma^{(1)} = \{r \in \sigma: \location{r} \in \setofpoints^{(1)}\} \text{ and } \sigma^{(2)} = \{r \in \sigma: \location{r} \in \setofpoints^{(2)}\}.
\end{equation*}
By definition of our algorithm, we rewrite $\algor(\sigma) = \algor_m(\sigma) + \algor_p(\sigma)$, where $\algor_m(\sigma)$ denotes the total connection + delay cost, and $\algor_p(\sigma)$ denotes the total penalty cost. 

We first upper bound $\algor_p(\sigma)$.
Note that the number of requests being cleared is at most $2 |\sigma^{(2)}|$. 
This is because, if an odd number of requests from $\sigma^{(2)}$ is cleared, then there exists one late request from $\sigma^{(1)}$ which is also cleared (i.e., the number of requests being cleared is $|\sigma^{(2)}| + 1 \le 2 |\sigma^{(2)}|$); otherwise, all the cleared requests are from $\sigma^{(2)}$. 
Since each request in $\sigma^{(2)}$ has its location's radius at least $p$, we have
\begin{equation*}
    \algor_p(\sigma) \le 2  |\sigma^{(2)}| \cdot p \le 2\sum_{r \in \sigma^{(2)}} \min\{\radius{\location{r}},p\}.
\end{equation*}
Thanks to Proposition \ref{proposition:minimum}, we have
\begin{equation} \label{penalty-ineq-1}
    \E_{\sigma}^m[\algor_p(\sigma)] \le 2m \sum_{x \in \setofpoints^{(2)}} \frac{\lambda_x}{\lambda(\setofpoints)} \cdot \min\{\radius{\location{r}}, p\},
\end{equation}
since each of the $m$ requests has a probability of $\sum_{x \in \setofpoints^{(2)}} \frac{\lambda_x}{\lambda(\setofpoints)}$ to appear in $\setofpoints^{(2)}$ (i.e., in $\sigma^{(2)}$).

Next, we upper bound $\algor_m(\sigma)$. 
By definition, when a request $r \in \sigma^{(2)}$ arrives, either it is matched with a pending request $r' \in \sigma^{(1)}$ or it is cleared immediately. 
This means no delay cost is incurred for $\sigma^{(2)}$. 
Besides, if $r \in \sigma^{(2)}$ is matched with $r' \in \sigma^{(1)}$ into a pair, we know that $\location{r} \in \overline{B}(\location{r'}, \radius{\location{r'}})$ (i.e., the connection cost of this pair is bounded by $\radius{\location{r'}}$) and $r'$ is a nice request.   
As a result, $\algor_m(\sigma)$ is produced from $\sigma^{(1)}$. 
Note that requests $\sigma^{(1)}$ are matched by applying our Radius algorithm. 
Same as the proofs of Claim \ref{claim:few_late_edges}, Claim \ref{claim:connection_cost_nice} and Claim \ref{claim:delay_cost_radius}, we have 
\begin{itemize}
    \item[-] the connection costs for matching all the late requests in $\sigma^{(1)}$ is bounded by $\frac{1}{2} \cdot |\setofpoints| \cdot d_{\max}$;
    \item[-] the connection cost for matching all the nice requests is bounded by $\sum_{r \in \sigma^{(1)}} \radius{\location{r}}$;
    \item[-] each request $r \in \sigma^{(1)}$ is delayed for a duration at most $Y_r$ (see $Y_r$'s definition above Claim \ref{claim:delay_cost_radius}).
\end{itemize}
As a result, thanks to the proof of Lemma \ref{lemma:radius_upperbound}, here we have
\begin{equation} \label{penalty-ineq-2}
    \E_{\sigma}^m[\algor_m(\sigma)] \le \left(2m \sum_{\anypoint \in \setofpoints^{(1)}} \frac{\lambda_x}{\lambda(\setofpoints)} \cdot \radius{\anypoint}\right) + \frac{1}{2} \cdot |\setofpoints| \cdot d_{\max}.
\end{equation}
By combining inequalities \eqref{penalty-ineq-1} with \eqref{penalty-ineq-2} together, we have
\begin{eqnarray}
    \E_{\sigma}^m[\algor(\sigma)] &=& \E_{\sigma}^m[\algor_m(\sigma)] + \E_{\sigma}^m[\algor_p(\sigma)] \nonumber \\
    &\le& \left(2m \sum_{\anypoint \in X} \frac{\lambda_x}{\lambda(\setofpoints)} \cdot \min\{\radius{\anypoint}, p\}\right) + \frac{1}{2} \cdot |\setofpoints| \cdot d_{\max}, \nonumber
\end{eqnarray}
which concludes the proof. 

\section{Missing proofs in Section \ref{section:asymmetric}}
\label{app:asymmetric}

\paragraph{Lower bounding $\E^m_\sigma[\optim(\sigma)]$.} In order to establish $\roe(\greedy) = 16/ (1 - e^{-2})$, we need to first re-define the radius for this asymmetric MPMD. 
\begin{definition}
    For each point $x \in \setofpoints$, define 
    \begin{eqnarray*}
        \overline{B}(x, u) &:=& \left\{y \in \setofpoints: \frac{d(x, y) + d(y, x)}{2} \le u \right\}, \\
        B^o(x, u) &:=& \left\{y \in \setofpoints: \frac{d(x, y) + d(y, x)}{2} < u\right\}.
    \end{eqnarray*}
\end{definition}

\begin{definition}
    For each point $x \in \setofpoints$, define its radius $\radius{x}$ as the smallest value $u \in \R^+ \cup \{\infty\}$ s.t.
    \begin{equation*}
        u \ge \frac{1}{\lambda(\overline{B}(x, u))}.
    \end{equation*}
\end{definition}

\begin{observation}
    For each point $x \in \setofpoints$, we have
    \begin{equation*}
        \frac{1}{\lambda(B^o(x, \rho_x))} \ge \rho_x \ge \frac{1}{\lambda(\overline{B}(x, \rho_x))}.
    \end{equation*}
\end{observation}

With the help of the radius definition, we again have
\begin{lemma}
    For asymmetric MPMD in the Poisson arrival model, the expected cost of the optimal offline solution, over all random sequences consisting of $m$ requests, satisfies
    \begin{equation*}
        \E_\sigma^m[\optim(\sigma)] \ge m \cdot \frac{1 - e^{-2}}{4} \sum_{x \in \setofpoints} \frac{\lambda_x}{\lambda(\setofpoints)} \cdot \radius{x}.
    \end{equation*}
\end{lemma}
Here comes the proof. 
Given any input sequence $\sigma$ and any request $r \in \sigma$, we define the minimum total cost of $r$ as
\begin{equation*}
    c(\sigma, r) := \min_{r' \in \sigma: r' \neq r} \left\{\frac{d(\ell(r), \ell(r')) + d(\ell(r'), \ell(r))}{2} + |t(r) - t(r')|\right\}.
\end{equation*}

\begin{claim}
    For any input sequence $\sigma$ it holds that $\optim(\sigma) \ge \frac{1}{2} \cdot \sum_{r \in \sigma} c(\sigma, r)$. 
\end{claim}
\begin{proof}
    For any pair $(r, r') \in \optim(\sigma)$, a cost of $$\frac{d(\ell(r), \ell(r')) + d(\ell(r'), \ell(r))}{2} + |t(r) - t(r')|$$ is incurred. 
    By definition of $c(\sigma, r)$, we have $$\frac{d(\ell(r), \ell(r')) + d(\ell(r'), \ell(r))}{2} + |t(r) - t(r')| \ge c(\sigma, r).$$ 
    Therefore, $2 \cdot \optim(\sigma) \ge \sum_{r \in \sigma} c(\sigma, r)$ and hence the claim.
\end{proof}

\begin{claim}
    Given a random sequence $\sigma$, we order the requests in $\sigma$ according to their arrivals. For any point $x \in \setofpoints$, the expected minimum cost of the $i$-th request $r_i$ in $\sigma$, assuming that $r_i$ is located at $x$, is
    \begin{equation*}
        \E_\sigma^m[c(\sigma, r_i) \mid \ell(r_i) = x] \ge \frac{1 - e^{-2}}{2} \cdot \radius{x}.
    \end{equation*}
\end{claim}
\begin{proof}
    We extend the random sequence $\sigma = (r_1, \dots, r_m)$ by infinite many random requests $r_j$ ($j \le 0$ and $j > m$) as
    \begin{equation*}
        \overline{\sigma} = (\dots, r_{-1}, r_0, r_1, \dots, r_m, \dots).
    \end{equation*}
    Note that for each integer $j$, $(t(r_{j+1}) - t(r_j)) \sim \expdistr{\lambda(\setofpoints)}$ and $\prob(\ell(r_j) = x) = \lambda_x / \lambda(\setofpoints)$. 
    This implies that in $\overline{\sigma}$, for any point $x \in \setofpoints$, with probability 1 there exist indexes $j \le 0$ and $j' > m$ satisfying $\ell(r_j) = \ell(r_{j'}) = x$.
    Furthermore, we have
    \begin{equation*}
        \E_\sigma^m[c(\sigma, r_i) \mid \ell(r_i) = x] \ge \E_{\overline{\sigma}}[c(\overline{\sigma}, r_i) \mid \ell(r_i) = x]. 
    \end{equation*}
    Note that for any $j, j' \in \{1, \dots, m\}$,
    \begin{equation*}
        \E_{\overline{\sigma}}[c(\overline{\sigma}, r_j) \mid \ell(r_j) = x] = \E_{\overline{\sigma}}[c(\overline{\sigma}, r_{j'}) \mid \ell(r_{j'}) = x] = \E_{\overline{\sigma}}[c(\overline{\sigma}, r_0) \mid \ell(r_0) = x].
    \end{equation*}
    As a result, now we only need to prove 
    \begin{equation*}
        \E_{\overline{\sigma}}[c(\overline{\sigma}, r_0) \mid \ell(r_0) = x] \ge \frac{1 - e^{-2}}{2} \cdot \rho_x.
    \end{equation*}
    Define 
    \begin{eqnarray*}
        W^- &:=& \min_{j < 0} \{-t(r_j): \frac{d(x, \ell(r_j)) + d(\ell(r_j), x)}{2} < \rho_x \} \ge 0 \\
        W^+ &:=& \max_{j > 0} \{t(r_j): \frac{d(x, \ell(r_j)) + d(\ell(r_j), x)}{2} < \rho_x\} \ge 0.
    \end{eqnarray*}
    We have
    \begin{eqnarray*}
        \lefteqn{c(\overline{\sigma}, r_0)} \\ 
        &=& \min_{j \neq 0} \left\{\frac{d(x, \ell(r_j)) + d(\ell(r_j), x)}{2} + |t(r_j)|\right\} \\
        &\ge& \min\left\{ \min_{j \neq 0} \Big\{\frac{d(x, \ell(r_j)) + d(\ell(r_j), x)}{2} + |t(r_j)|\Big\}, \rho_x \right\} \\
        &=& \min\left\{ \min_{j < 0} \Big\{\frac{d(x, \ell(r_j)) + d(\ell(r_j), x)}{2} + |t(r_j)|\Big\}, \min_{j > 0} \Big\{\frac{d(x, \ell(r_j)) + d(\ell(r_j), x)}{2} + |t(r_j)|\Big\}, \rho_x \right\} \\
        &=& \min\{\min\{W^-, W^+\}, \rho_x\}.
    \end{eqnarray*}
    Again, $W^-$ and $W^+$ are mutually independent and follow the same exponential distribution $\expdistr{\lambda(B^o(x, \rho_x))}$. 
    We thus have $\min\{W^-, W^+\} \sim \expdistr{2\lambda(B^o(x, \rho_x))}$.
    Thanks to Claim \ref{claim:exponential-2}, we have this claim and hence the lower bounding scheme.
\end{proof}

\paragraph{Upper bounding the cost of Greedy algorithm. } For this asymmetric distance case, the greedy algorithm shall match two requests $r, r'$ into a pair when their total delay cost is at least $\frac{d(\ell(r), \ell(r')) + d(\ell(r'), \ell(r))}{2}$. 

\RestyleAlgo{boxruled}
\LinesNumbered
\SetAlgoVlined
\begin{algorithm}
\caption{Greedy}
\label{pseudocode:greedy_asymmetric}
\KwIn{A sequence $\sigma$ of requests.}
\KwOut{A perfect matching of the requests.}
\For{any time $t$}{
    \If{
        there exist pending requests $\request, \request'$ such that $(t - t(\request)) + (t - t(\request')) \ge \Big(d(\ell(r), \ell(r')) + d(\ell(r'), \ell(r))\Big) /2$
    }{
        match them into a pair with ties broken arbitrarily.
    }
}
\end{algorithm}
\vspace{-6pt}

Again, denoting by $w(r)$ the waiting time of a request $r \in \sigma$, we have
\begin{equation*}
    \greedy(\sigma) \le 2 \sum_{r \in \sigma} w(r).
\end{equation*}

We now focus on bounding the waiting time of each request. 
To do this, we distinguish two types of requests. 
For each request $r$, define $t'(r) := t(r) + \rho_{\ell(r)}$.
We say that $r$ is a \emph{late} request if 
\begin{itemize}
    \item[-] $r$ is still pending at time $t'(r)$ and
    \item[-] there is no request $r'$ arriving within the closed ball of $r$'s location (i.e., $d(\ell(r), \ell(r')) \le \rho_{\ell(r)}$) after time $t'(r)$.
\end{itemize}
Otherwise, we say that $r$ is a \emph{nice} request, and we define
\begin{equation*}
    Y^\text{nice}_r :=
    \begin{cases}
    0 & \textit{ if } r \textit{ is matched at time } t'(r);\\
    \displaystyle\min_{r'\in\sigma} \left\{t(r') - t'(r) \mid t(r') > t'(r) \text{ and } d(\ell(r'),\ell(r)) \le \rho_{\ell(r)}\right\} & \textit{ otherwise.}
    \end{cases}
\end{equation*}
We bound the waiting time of nice requests as follows:
\begin{claim}
    For each nice request $r \in \sigma$, we have $w(r) \le \rho_{\ell(r)} + Y^\text{nice}_r$.
\end{claim}

We now bound the total delay time induced by late requests. 
Unfortunately, the waiting time of a late request can possibly be as large as the diameter $d_{\max} = \max_{x, y \in \setofpoints} d(x,y)$ of the metric space. 
However, we show that there are only constantly many such requests. 
Let $t(r_m)$ denote the arrival time of the last request in $\sigma$. 
For any late request $r$, define 
\begin{equation*}
    Y^\text{late}_r :=
    \begin{cases}
    0 & \textit{ if } t(r) + d_{\max} \ge t(r_m);\\
    \displaystyle\min_{r' \in \sigma} \left\{t(r') - (t(r) + d_{\max}) \mid t(r') > t(r) + d_{\max}\right\} & \textit{ otherwise.}
    \end{cases}
\end{equation*}

\begin{claim} 
    For any point $x\in \setofpoints$, there is at most one late request located on $x$. 
    In particular, there are at most $|\setofpoints|$ late requests. 
    Moreover, for each late request $r$, we have $w(r) \le d_{\max} + Y^\text{late}_r$. 
\end{claim}

We now use stochastic assumptions to upper bound the expected cost of the solution. 
\begin{lemma}
For asymmetric MPMD in the Poisson arrival model, the expected cost produced by the Greedy algorithm, over all random sequences consisting of $m$ requests, satisfies
\begin{equation*}
    \E_{\sigma}^m[\greedy(\sigma)] \le \left( 4m \sum_{\anypoint \in \setofpoints} \frac{\lambda_x}{\lambda(\setofpoints)} \cdot \radius{\anypoint}\right)+2|\setofpoints| \cdot \left(d_{\max}+\frac{1}{\lambda(\setofpoints)}\right).
\end{equation*}
where $d_{\max} := \displaystyle\max_{x,y \in \setofpoints} d(x,y)$ is the diameter of the metric space. 
\end{lemma}

\begin{proof}
Let $\sigma = (r_1, \dots, r_m)$ be a random sequence of $m$ requests, where requests are ordered with increasing arrival times. 

We first bound the expected delay cost induced by late requests. 
Suppose that the $i$-th request of the sequence, $r_i$, is late and is located on a point $x\in \setofpoints$. 
Using Definition \ref{definition:centralized_poisson_arrival}, we know that when $\sigma$ is a random sequence generated by the Poisson arrival process, the (conditional) random variable $Y^\text{late}_r$ follows an exponential distribution of parameter $\lambda(\setofpoints)$. 
In particular, we obtain
\begin{equation*}
    \E_\sigma^m[w(r_i) \mid r_i \text{ is late and }\ell(r_i) = x] \le d_{\max} + \frac{1}{\lambda(\setofpoints)}. 
\end{equation*}
Since the expectation does not depend on point $x$, and since there are at most $|\setofpoints|$ late requests, the total delay cost induced by the late requests is in expectation:
\begin{equation*}
    \E_\sigma^m\left[\sum_{\substack{i=1 \\r_i\text{ is late}}}^m w({r_i})\right]\le |\setofpoints| \cdot \left(d_{\max} + \frac{1}{\lambda(\setofpoints)}\right).
\end{equation*}
When $r_i$ is a nice request located on $x \in \setofpoints$, by Proposition \ref{proposition:subset_exponential_waiting_time}, the (conditional) random variable $Y^\text{nice}_r$ follows an exponential distribution of parameter 
\begin{equation*}
    \sum_{y \in \setofpoints: \frac{d(x,y)+d(y,x)}{2} \le \rho_x} \lambda_y = \lambda(\overline{B}(x,\rho_x)).
\end{equation*}
Using Observation \ref{observation:lowerboundingscheme} we obtain:
\begin{equation*}
    \E_\sigma^m[w(r_i) \mid r_i \text{ is nice and } \ell(r_i) = x] = \rho_x + 1 / \lambda(\overline{B}(x,\rho_x)) \le \rho_x + \rho_x = 2\rho_x.
\end{equation*}
As we observed in Definition \ref{definition:centralized_poisson_arrival}, the probability the the $i$-th request of the (random) sequence is located on $x$ is equal to $\lambda_x / \lambda(\setofpoints)$. 
Thus, the total delay cost induced by nice requests is 
\begin{align*}
    \E_\sigma^m\left[\sum_{\substack{i=1 \\r_i\text{ is nice}}}^m w({r_i})\right]
    &\le \sum_{i=1}^m \sum_{x \in \setofpoints} \prob(r_i \text{ is nice and }\ell(r_i) = x) \cdot \E_\sigma^m\left[w({r_i}) \mid r_i \text{ is nice and }\ell(r_i) = x \right]\\
    &\le \sum_{i=1}^m \sum_{x \in \setofpoints} \prob(\ell(r_i)=x) \cdot 2\rho_x\\
    &\le \sum_{i=1}^m \sum_{x \in \setofpoints} \frac{\lambda_x}{\lambda(\setofpoints)} \cdot 2\rho_x\\
    &= m \sum_{x \in \setofpoints} \frac{\lambda_x}{\lambda(\setofpoints)} \cdot 2\rho_x.
\end{align*}
Finally, putting everything together, we obtain the expected bound:
\begin{align*}
    \E_\sigma^m\left[\greedy(\sigma)\right]
    &\le \E_\sigma^m\left[2\sum_{i=1}^m w({r_i})\right] &&\text{(Claim \ref{claim:greedy_connected_at_most_delay})}\\
    &= 2 \cdot \E_\sigma^m\left[\sum_{\substack{i=1 \\r_i\text{ is nice}}}^m w({r_i})\right] + 2 \cdot \E_\sigma^m\left[\sum_{\substack{i=1 \\r_i\text{ is late}}}^m w({r_i})\right] \\
    &\le \left(4m \sum_{x\in\setofpoints}\frac{\lambda_x}{\lambda(\setofpoints)}\cdot\rho_x \right) + 2 |\setofpoints| \cdot \left(d_{\max} + \frac{1}{\lambda(\setofpoints)}\right).
\end{align*}
This concludes the proof. 
\end{proof}

\end{document}